\documentclass{article}
\usepackage{tikz,a4wide,amsthm}

\usepackage[colorlinks]{hyperref}

\usepackage{graphics,graphicx,color,url}      % pictures
\graphicspath{{Figures/}}
\usepackage{amsmath,amssymb,graphicx,bbm,latexsym}
\usepackage{cite} % sorts citations by numbers

\usepackage[algoruled,vlined,linesnumbered]{algorithm2e}
\SetKwComment{tcp}{//}{}
\newtheorem{algo}{Algorithm}

\usepackage{tabularx}
\usepackage{paralist} % compact list environments
\usepackage{pdfpages}

\usepackage[utf8]{inputenc}
\usepackage[T1]{fontenc}
\usepackage{url}

\newcommand{\TestL}[1]{\begin{tikzpicture}[scale=0.6]\useasboundingbox (0,0) rectangle (1,1);\draw (0,0) rectangle (1,1) (0.5,0.5) node{#1} ;\end{tikzpicture}}
\newcommand{\Test}{\begin{tikzpicture}[scale=0.6]\useasboundingbox (0,0) rectangle (1,1);\draw (0,0) rectangle (1,1);\end{tikzpicture}}

\newcommand{\ExecL}[2]{\begin{tikzpicture}[scale=0.6]\useasboundingbox (0,0) rectangle (#1,1);\draw [fill=gray!40!white] (0,0) rectangle (#1,1) (#1/2,0.5) node{#2} ; \draw[ultra thick](#1,-0.03) -- (#1,1.03);\end{tikzpicture}}

\theoremstyle{plain}
\newtheorem{theorem}{Theorem}
\newtheorem{lemma}[theorem]{Lemma}

\newtheorem{corollary}[theorem]{Corollary}

\newtheorem{proposition}[theorem]{Proposition}

\newcommand{\keywords}[1]{} % do nothing

\newcommand{\T}{$T_1$\xspace}
\newcommand{\Ex}{E\xspace}%\max\{1, \ub - 1\}
\newcommand{\tbeatthresh}{$T_2$\xspace}
\newcommand{\ub}{\bar p\xspace}

\newcommand{\p}{\ub}%{\bar{p}}
\newcommand{\ratio}{\rho_{\infty}}
\newcommand{\mratio}{\rho}

\usepackage{xspace}
\newcommand{\algoname}[1]{\textsc{#1}\xspace}

\newcommand{\OPT}{\algoname{Opt}}
\newcommand{\ALG}{\algoname{Alg}}

\newcommand{\BOUND}{\algoname{Threshold}}
\newcommand{\RTE}{\algoname{Random}}
\newcommand{\thresh}{\algoname{Threshold}}
\newcommand{\beat}{\algoname{Beat}}
\newcommand{\ute}{\algoname{UTE}}
\newcommand{\TA}{\algoname{DelayAll}}

\newcommand{\eps}{\epsilon}

\usepackage[normalem]{ulem} %% for \sout = strikethrough
\title{An Adversarial Model for Scheduling with Testing
\thanks{This research was carried out in the framework of \textsc{Matheon} supported by Einstein Foundation Berlin,  the German Science Foundation (DFG) under contract  ME~3825/1 and Bayerisch-Französisches Hochschulzentrum (BFHZ). Further support was provided by EPSRC grant EP/S033483/1 and the ANR grant ANR-18-CE25-0008. The second author was supported by a study leave granted by University of Leicester during the early stages of the research.
A preliminary version of this paper appeared in \emph{The 9th Innovations in Theoretical Computer Science Conference} (ITCS), January 2018~\cite{DurrEMM/18}.}
}

\author{Christoph D\"urr\thanks{Sorbonne Universit\'e, CNRS, Laboratoire d'informatique de Paris 6, Paris, France.}
                    \and
Thomas Erlebach\thanks{School of Informatics, University of Leicester, UK.}
                    \and
Nicole Megow\thanks{Department of Mathematics and Computer
  Science, University of Bremen, Germany.}
                    \and
Julie Mei\ss{}ner\thanks{Institute of Mathematics, Technical University of Berlin, Germany.}
}

\begin{document}
\maketitle

%!TEX root = main.tex

\begin{abstract}
%We consider a novel %single-machine
%scheduling model where
\noindent We introduce a novel adversarial model for scheduling with explorable uncertainty. In this model, the processing time of a job can potentially be reduced (by an \emph{a priori} unknown amount) by testing the job. Testing a job~$j$ takes one unit of time and may reduce its processing time from the given upper limit $\bar{p}_j$ (which is the time taken to execute the job if it is not tested) to any value between $0$ and~$\bar{p}_j$. This setting is motivated e.g.\ by applications where a code optimizer can be run on a job before executing it. We consider the objective of minimizing the sum of completion times on a single machine. All jobs are available from the start, but the reduction in their processing times as a result of testing is unknown, making this an online problem that is amenable to competitive analysis. The need to balance the time spent on tests and the time spent on job executions adds a novel flavor to the problem. We give the first and nearly tight lower and upper bounds on the competitive ratio for deterministic and randomized algorithms. We also show that minimizing the makespan is a considerably easier problem for which we give optimal deterministic and randomized online  algorithms.
%We give a $2$-competitive deterministic algorithm and prove a lower bound of $1.8546$ on the best possible competitive ratio of any deterministic algorithm. Furthermore, we show that randomization helps and present a $1.7453$-competitive randomized algorithm. We also give a lower bound of $1.6257$ on the best possible competitive ratio of any randomized algorithm. Both lower bounds hold even for instances with uniform upper limits where each processing time is either $0$ or the upper limit. For this special case, we give a deterministic $1.8668$-competitive algorithm. For the case of uniform upper limits and arbitrary processing times, we give a deterministic $1.9338$-competitive algorithm.
\end{abstract}

\section{Introduction}

Uncertainty in scheduling has been modeled and investigated in many different ways, particularly in the frameworks of online optimization, stochastic optimization, and robust optimization. All these different approaches have the common assumption that the uncertain information, e.g., the processing time of a job, cannot be explored before making scheduling decisions. However, in many applications there is the opportunity to gain exact or more precise information at a certain additional cost, e.g., by investing time, money, or energy. It is a challenging problem to design algorithms that balance the cost for data exploration and the benefit for the quality of a solution. This involves quantifying the trade-off between exploration and exploitation, as it is ubiquitous in numerous applications.

In this paper, we introduce a novel model for scheduling with explorable uncertainty. Given a set of $n$ jobs, every job $j$ can optionally be \emph{tested} prior to its execution. A job that is executed without testing has processing time %\footnote{We define the problem with rational numbers for the ease of representing them in a computer, \te{but all our results and proofs also hold for real numbers}.} 
$\bar{p}_j\geq 0% \in \mathbbm{Q}^+
$, while a tested job has processing time $p_j$ with $0 \le p_j \le \bar{p}_j$. Testing a job takes one unit of time on the same resource (machine) that processes jobs.  A tested job does not need to be executed right after its test.

%\te{The ability to preempt the execution of a test or of a (tested or untested) job would be of no benefit to the algorithm or the adversary
%as no new information is obtained during the execution. Therefore, we only consider algorithms and schedules that do not preempt tests and that do not preempt job executions. However, as noted above, the execution of a tested job can be scheduled any time after the completion of the test.}

Initially the algorithm knows for each job $j$ only the upper limit $\bar{p}_j$, and gets to know the time $p_j$ only after a test. Tested jobs can be executed at any time after their test. An algorithm must carefully balance testing and execution of jobs by evaluating the benefit and cost for testing.  The resulting schedule is constructed by the algorithm adaptively. This means that at every moment, the choice of a job, and the decision whether to execute or to test it, may depend on the outcome of previous tests.

%We investigate single-machine scheduling problem with $n$ jobs.
%%, where the goal is to minimize the sum of completion times.

We focus on scheduling on a single machine. Unless otherwise noted, we consider the sum of completion times as the minimization objective. We use competitive analysis to assess the performance
of algorithms.

For the standard version of this single-machine scheduling problem, i.e., without testing, it is well known that the Shortest Processing Time (\algoname{SPT}) rule is optimal for minimizing the sum of completion times. The addition of testing, combined with the fact that the processing times $p_j$ are
initially unknown to the algorithm, turns the problem into an online
problem with a novel flavor.
An algorithm must decide which jobs to execute untested and which
jobs to test. Once a job has been tested, the algorithm must decide whether
to execute it immediately or to defer its execution while testing or
executing other jobs. At any point in the schedule, it may be difficult
to choose between testing a job (which might reveal that it has a very short
processing time and hence is ideally suited for immediate execution) and
executing an untested or previously tested job. Testing a job yields
information that may be useful for the scheduler, but may delay the completion
times of many jobs. Finding the right balance between tests and executions
poses an interesting challenge.

%If the processing times $p_j$ that jobs have after testing are known,
%an optimal schedule is easy to determine: Testing and executing
%job $j$ takes time $1+p_j$, so it is beneficial to test the job only
%if $1+p_j<\bar p_j$. In the optimal schedule, jobs are therefore ordered by
%non-decreasing $\min \{1 + p_j, \bar p_j \}$. In this order, the jobs
%with $1+p_j<\bar p_j$ are tested and executed while jobs with $1+p_j\ge\bar p_j$
%are executed untested. (For jobs with $1+p_j=\bar p_j$ it does not matter
%how they are processed.)
% xtof: simpler?
% whether the job is tested and executed, or executed untested.)

%At any time, we denote by $\hat p_j$ the `current' running-time of job~$j$,
%i.e., $\hat p_j = \bar{p}_j$ if $j$ has not yet been tested and
%$\hat p_j = p_j$ if $j$ has been tested. By $p_j^*=\min\{p+p_j,\bar{p}_j\}$ we denote the execution time for job~$j$
%(including testing if testing is beneficial for the job) in the optimal solution. \julie{Both notations $\hat p_j$ and $p^*_j$ are rarely used. Maybe we should just delete them for the general part.}

\subsection{Motivation and applications}

Scheduling with testing is motivated by a range of application settings
where %an operation that corresponds to 
an information-revealing test can be applied to jobs before they are executed which leads to a trade-off regarding how to allocate resources for performing a test and actually executing the job. We discuss some examples of such settings from very different domains. %here.

First, consider the execution of computer programs on a processor. A test
could correspond to a code optimizer that takes unit time to process the
program and  potentially reduces its running-time. The upper limit of a job
describes the running-time of the program if the code optimizer is not
executed.  See \cite[Chapter 5]{cardoso2017embedded} for an overview of
various code otimization techniques.

Second, consider the transmission of files over a network link. It is possible
to run a compression algorithm that can reduce the size of a file by an
\emph{a priori} unknown amount. If a file is incompressible (e.g., if it is
already compressed), its size cannot be reduced at all. Running the
compression algorithm corresponds to a test.  See
\cite{rothstein2013adaptive,wiseman2005efficient} for some practical
techniques balancing compression time with transmission time.

%A third application comes from spray painting objects. The solvent used in
%those paints can attack the surface of the object depending on its material.
%Therefore, for each object we have the option either to apply a protective
%coating layer followed by spray painting, or to omit the time consuming
%coating, which requires a test on a small portion of the object.\nic{do we find a reference for this?}

An algorithmic application concerns jobs, which can be executed in two
different modes, a \emph{safe} mode and an \emph{alternative} mode. The safe
mode is always possible. The alternative mode may have a shorter processing
time, but is not possible for every job. A test is necessary to determine
whether the alternative mode is possible for a job and what the processing
time in the alternative mode would be.  One example would be computing
shortest paths in several given graphs. Solving it in the safe mode would
involve the Bellman-Ford algorithm, while the faster alternative mode uses
Dijkstra's algorithm requiring a preliminary non-negativity test on the edge
weights. This situation would be faced by a server that
solves shortest paths problems submitted by users.
See \cite{kerschke2019automated} for a survey on algorithm selection
techniques.

% Suppose we have to compute a maximum flow for a several given graphs $G(V,E)$. On each instance, we could run a standard flow algorithm (the \emph{safe} mode), for example Dinic's algorithm, running in time $O(|V|^2|E|)$. But if the graph is planar, we can first compute in time $O(|V|+|E|)$ a planar embedding of the graph, and then run Dijkstra's shortest path algorithm (the \emph{alternative} mode)on the dual of the planar embedding of the graph, in time $O(|E|+|V|\log|V|)$. 

%As a final application area, consider settings where a diagnosis can be carried out to determine the exact rocessing time of a job. For example, fault diagnosis can determine the time needed for a repair job, or a medical diagnosis can determine the time needed for a consultation and treatment session with a patient. Assume that the resource that carries out the diagnosis is the same resource that executes the job (e.g., an engineer or a medical doctor), and that the resource must be allocated to a job for an uninterruptible period that is guaranteed to cover the actual time needed for the job. If the diagnosis takes unit time, we arrive at our problem of scheduling with testing.

As a final application area consider scenarios, where a diagnosis can be carried out to determine the exact processing time of a job. This is the case in very different domains such as diagnostics in {\em maintenance} or in {\em medical environments} such as emergency departments. A fault diagnosis can determine the time needed to repair or replace a device, which allows for an efficient schedule of maintenance operations. There is a vast amount of literature on maintenance models; see e.g. \cite{Nicolai2008,PierskalleV76}. In medical diagnostics, %e.g. in emergency departments, 
information can be acquired about the time needed for consultation, treatment session and other activities with the patient. This information can help to prioritize and efficiently allocate limited medical resources; cf.~\cite{LeviMS19,AlizamirVS13,MillsAZ13}. Assuming that the resource that performs the diagnosis is the same resource that executes the job, e.g., an engineer or a medical doctor, we are in our problem setting of scheduling with testing with the trade-off regarding how to allocate resources between diagnostics and actual execution of jobs.
%and that the resource must be allocated to a job for an uninterruptible period that is guaranteed to cover the actual time needed for the job. If the diagnosis takes unit time, we arrive at our problem of scheduling with testing.

% \begin{description}
%   \item[processing programs on a machine] It is possible to run during a unit
%   time step a code optimizer that might reduce the running time of the job.
%
%   \item[transmission of files] It is possible to run a compression algorithm,
%   which could fail if the file is uncompressable (for example if the file is
%   already compressed).
%
%   \item[safe and alternative method to process a job] A job $j$ can be
%   processed either under the safe method which takes time $\bar{p}_j$ or under
%   an alternative method which takes time $p_j$. The alternative method is not
%   always possible, and a preliminary test needs to be done first.
% \end{description}
% For these motivations one might argue that the test time could be linear in
% the upper bounds, so we have to explain that it is dominated by the fixed part
% in the test time.

In some applications, it may be appropriate to allow the time for testing a job
to be different for different jobs (e.g., proportional to the upper limit of a
job).  Furthermore, there are applications where the job processing time is $p_j$ even if executed untested, and the test reveals $p_j$, which otherwise is only known to belong to the interval $[0,\bar p_j]$.
We leave the consideration of such generalizations of the problem to
future work.

\subsection{Our contribution}
% \new{A scheduling algorithm in the model of explorable uncertainty has to make two types of decisions: which jobs are to be tested, and in what order job executions and tests should be interleaved.  There is a subtle compromise to be found between investing time to test jobs and the benefit one can obtain from the tests.  If\nic{Den Satz finde ich schwierig (misleading/falsch).} the goal were to minimize the sum of completion times then no optimal algorithm would ever take the risk to test jobs.  But since we are concerned with minimizing the competitive ratio, testing a job could either decrease the cost of the algorithm or increase the worst-case optimal schedule.  This is an interesting aspect of our model.  We propose different algorithms for several variants of the problem.  For the analysis, first we show that worst-case instances
% have a particular structure that can be described by a few parameters.  Then we express both the cost of the algorithm and the cost of the optimal schedule as functions of these parameters.  Second we analyze the worst-case ratio among all possible parameter values. This is the main technical part, which involves second order analysis for which we sometimes use the help of the algebraic solver Mathematica.}\julie{Ich bevorzuge die grüne Variante, aber sollten wir uns für blau entscheiden würde ich diesen Satz abändern. 'we employ the algebraic solver Mathematica whenever necessary'}

A scheduling algorithm in the model of explorable uncertainty has to make two types of decisions: which jobs should   be tested, and in what order should job executions and tests be scheduled.  There is a subtle compromise to be found between investing time to test jobs and the benefit one can gain from these tests. We design scheduling algorithms that address this exploration-exploitation question in different ways and provide nearly tight bounds on the competitive ratio. In our analysis, we  first show that worst-case instances have a particular structure that can be described by only a few parameters. This goes hand in hand with analyzing also the structure of both an optimal and an algorithm's schedule. Then we express the total cost of both schedules as functions of these few parameters. It is noteworthy that, under the assumptions made, we typically characterize the {\em exact} worst-case ratios of the considered algorithms.  
Given the parameterized cost ratio, we analyze the worst-case parameter choice. This technical part  involves second order analysis which we perform with computer assistance. 
These computations are provided as notebook- and pdf-files
at a companion webpage.\footnote{%
Files that can be opened with the algebraic solver \emph{Mathematica}
are available at the URL \url{http://cslog.uni-bremen.de/nmegow/public/mathematica-SwT.zip}.}

Two variants of the problem attracted our attention in particular.  In an
\emph{uniform} instance all jobs have the same upper limit $\bar p_j$, which
makes them initially undistinguishable to the scheduler.  This means that the
algorithm's decision whether to test a job, does not depend on the job itself,
but only on the outcome of previous tests.  Moreover in an \emph{extreme
uniform instance}, after testing a job $j$ its processing time is either $0$
or $\bar p_j$.  This means that the benefit of a test is either maximized or
none at all.  Intuitively one would think that extreme uniform instances
capture the worst case instances of the problem, hence it is not surprising
that our lower bound constructions are of this form.   In addition we design
specific algorithms for these variants.  One motivation was to follow a detour
in order to find a better deterministic algorithm for the general problem,
which unfortunately failed. Another motivation is that there is a huge amount of
literature for scheduling problems with equal processing time. Therefore we
believe that these variants are interesting for their own.

\begin{table}[tb]
    \begin{center}
    \begin{tabular}{l|l|l@{\:}l}
        competitive ratio & lower bounds & \multicolumn{2}{l}{upper bounds}
        \\ \hline & \\[-0.8em]
        deterministic algorithms & 1.8546 (Thm \ref{thm:detLB}) & 2 &\textsc{Threshold} (Thm \ref{th:bound2}) \\
        randomized algorithms & 1.6257 (Thm \ref{thm:randLB}) & 1.7453 &\textsc{Random} (Thm \ref{thm:RTE}) \\
        uniform instances (det) & 1.8546 (Thm \ref{thm:detLB}) & 1.9338* &\textsc{BEAT} (Thm \ref{thm:uniformUB}) \\
        extreme uniform instances (det) & 1.8546 (Thm \ref{thm:detLB}) & 1.8668 &\textsc{UTE} (Thm \ref{thm:ext-uniform-UB})\\
        extreme uniform with %$\p\approx1.9896$ 
        $\p\approx1.989$ (det) & 1.8546 (Thm \ref{thm:detLB}) & 1.8552 &\textsc{UTE} (Cor \ref{cor:ute}) \\
    \end{tabular}
    \end{center}
    \caption{Our results for minimizing the sum of completion times.
    * holds asymptotically}
    %* the announced ratio of \textsc{BEAT} holds only asymptotically.}
    \label{table:contributions}
\end{table}

Our results are the following. For scheduling with testing on a single machine
with the objective of minimizing the sum of completion times, we present a
$2$-competitive deterministic algorithm and prove that no deterministic
algorithm can achieve competitive ratio less than $1.8546$. We then present a
$1.7453$-competitive randomized algorithm, showing that randomization provably
helps for this problem. We also give a lower bound of $1.626$ on the best
possible competitive ratio of any randomized algorithm. Both lower bounds hold
even for \emph{extreme uniform instances}, i.e.\ instances with uniform upper
limits where every processing time is either $0$ or equal to the upper limit.
For such instances we give a $1.8668$-competitive algorithm. In the special
case where the upper limit of all jobs is $\approx 1.9896$, the value used in
our deterministic lower bound construction, that algorithm is even
$1.8552$-competitive, which is nearly optimal. For the case of uniform upper limits and arbitrary
processing times, we give a deterministic $1.9338$-competitive algorithm. An
overview of these results is shown in Table~\ref{table:contributions}.

Finally, we give %in Section~\ref{sec:makespan} 
tight results for the simpler problem of minimizing
the makespan in scheduling with testing. The best possible deterministic algorithm has competitive ratio $\varphi \approx 1.618$, where $\varphi$ is the Golden ratio. The optimal randomized algorithm has competitive ratio $4/3$.

In the problem that we introduce in this paper, the interplay between the 
online algorithm and the adversary has a novel flavor due to the presence
of tests:
%In this paper, we introduce a problem, in which minimizing the objective value and
%minimizing the competitive ratio are two different goals for an algorithm.
Testing a job forces the adversary to select a specific
processing time right away, while otherwise the adversary can make this choice \emph{after} the
algorithm has completed all jobs.  To our knowledge, this kind of interaction
does not appear in the standard online computation framework. %We find it interesting,
%but were also surprised by how hard it is to analyze the competitive ratio of the problem.}\nic{Wollen wir den letzen Satz einfach streichen?}

From a technical perspective, our contribution consists of two parts. First we
present techniques to modify instances in an adversarial manner, while
reducing the number of distinct job parameters. %, e.g.\ Lemma~\ref{lem:linear-pj}.
This allows us to describe the
competitive ratio with a few parameters. Second we show how second order
analysis can be used to optimize these parameters.   

\paragraph{Organization of the paper.} In Section~\ref{sec:preliminaries}, we give the problem definition, some observations and structural properties. Section~\ref{sec:det} is devoted to lower and upper bounds for deterministic algorithms for general instances for minimizing the sum of completion times. Section~\ref{sec:rand} addresses randomized algorithms. In Section~\ref{sec:uniform}, we give more fine-grained results for special cases of the problem with uniform upper bounds. Finally, in Section~\ref{sec:makespan} we give optimal deterministic and randomized algorithms for minimizing the makespan.

\subsection{Related work} 

The arguably most classical framework modeling sequential decision making problems %under limited information 
	with an exploration-exploitation trade-off is the stochastic multi-armed bandit problem. In each round, one choses from a set of actions (bandit arms) %at a fixed cost 
	and obtains some observable payoff, where the goal is to maximize the total payoff. Since its introduction in 1933 in~\cite{Thompson33} a plethora of variants has been analyzed and till today this is an actively studied area with applications particularly in online auctions, adaptive routing, and the theory of learning in games; see e.g. \cite{GittinsGW11-book,BubeckC12}.

One of the oldest stochastic problems with explicit exploration cost is Weitzman's Pandora's box problem \cite{Weitzman1979}. Given $n$ random variables with probability distributions, the goal is to find a single variable of largest value, but one needs to pay a cost for each probe of a variable. Its solution can be stated as a special case of the Gittins index theorem~\cite{gittins74,KleinbergWW16}. A nice exposition of an application of a variation of the Gittins index to a problem that can be stated as `playing golf with two balls' can be found in~\cite{Dumitriu2003}.
%A large \te{amount of} work has been conducted on stochastic models.  A classic example is Weizmann's ``Pandora's Box,'' where one is given $n$ random variables and needs to choose a single variable with a high value, but needs to pay a cost for each probe of a variable.  
%The paradigm of optimizing the sum of query costs and objective value of the selected solution has been investigated more recently for various classic combinatorial otimization problems, such as, for example, the maximum profit matching or the knapsack problem, see \cite{singla2018price} and references therein. 
%For example, in a recent paper \cite{gupta2019markovian}, a combinatorial optimization model is studied where one needs to produce a solution to some combinatorial optimization problem, each decision variable having a hidden weight, which gets revealed in costly stages following a Markov chain.
Only recently, combinatorial otimization problems have been studied in this context with the goal of optimizing the sum of query costs and the objective value of the selected solution. This includes problems such as matching, set cover, facility location, and prize-collecting Steiner tree; see, e.g., \cite{singla2018price,gupta2019markovian} and references therein, also with uncertainty in the cost function~\cite{YamaguchiM19}. 
 
 Other stochastic problems taking exploration cost into account, such as stochastic knapsack~\cite{DeanGV08,Ma18}, orienteering~\cite{GuptaKNR15,BansalN15}, matching~\cite{ChenIKMR09,BansalGLMNR12,BlumDHPSS20,BehnezhadFHR19,AssadiKL19}, and probing problems \cite{AdamczykSW16,GuptaN13,GuptaNS16}, employ a {\em query-commit} model, which means that queried elements must be part of the solution, or it is required that the solution elements are queried. These are quite strong restrictions which change the nature of the  benefit-cost trade-off that an algorithm experiences when making queries.

 All these models have in common %that the problem at hand obeys some stochastic model.  
 that the uncertain information follows some stochastic model.
 We follow a more pessimistic approach, by studying an online or robustness model where the algorithm has no prior stochastic information. As usual in the absence of a known distribution, we assume the worst case and let an adversary chose the hidden information. 
 
This adversarial model falls in the area of deterministic explorable~(or queryable) uncertainty, where additional information about the input can be learned using a query operation, a test in our setting. 
The line of research on optimization with explorable uncertain data has been initiated by Kahan~\cite{Kahan91} in 1991. His work concerns
selection problems with the goal of minimizing the number of queries that are necessary to find the optimal solution.
After the initiation by Kahan~\cite{Kahan91} on selection problems. further problems have been studied in this uncertainty model including finding the~$k$-th smallest value in a set of uncertainty intervals~\cite{Kahan91,GuptaSabharwal:16:The-update,FederMPOW03} (also with non-uniform query cost~\cite{FederMPOW03}), caching problems in distributed databases~\cite{OlstonW00}, computing a function value~\cite{KhannaT01}, and classical combinatorial optimization problems, such as shortest path~\cite{FederMOCOP2007}, finding the median~\cite{FederMPOW03}, the knapsack problem~\cite{GoerigkGISS15}, and the MST problem~\cite{ErlebachHKMR08,MegowMeisner:17:Randomization,FockeMM17}.
While most work aims for minimal query sets to guarantee exact optimal solutions, Olsten and Widom~\cite{OlstonW00} initiate the study of trade-offs between the number of queries and the precision of the found solution. They are concerned with caching problems. Further work in this vein can be found in~\cite{KhannaT01,FederMOCOP2007,FederMPOW03}.

In all this previous work, the execution of queries is separate from the actual optimization problem being solved. In our case, the tests
are executed by the same machine that runs the jobs. Hence, the tests are not considered separately, but they directly affect
the objective value of the actual problem (by delaying the completion of other jobs while a job is being tested). Therefore, instead of minimizing the number of tests needed until an optimal schedule can be computed (which would correspond to the standard approach in the work on explorable uncertainty discussed above), in our case the tests of jobs are part of the schedule, %being produced,
and we are interested in the sum of completion times as the single objective function.

Our adversarial model is inspired by (and draws motivation from) recent work on a stochastic model of
scheduling with testing introduced by Levi, Magnanti and Shaposhnik~\cite{LeviMS19,Shaposhnik16}.
They consider the problem of minimizing the weighted sum of completion times on
one machine for jobs whose processing times and weights are random variables
with a joint distribution, and are independent and identically distributed across
jobs. In their model, testing a job does not make its processing time shorter,
it only provides information for the scheduler (by revealing the exact weight
and processing time for a job, whereas initially only the distribution is
known). They present structural results about optimal policies and
efficient optimal or near-optimal solutions based on dynamic
programming.

Scheduling problems, in general, have been studied extensively over decades. They occur in many different variations in a wide range of applications ranging from traditional production scheduling and project planning to new  resource management tasks arising in the advent of internet technology such as distributed cloud service networks and the allocation or virtual machines to physical servers. For a general overview and classification, we refer to the reference works \cite{Leung04-handbook-scheduling,Pinedo-book}. 

The most common frameworks for modeling scheduling with uncertain input are stochastic scheduling \cite{moehringRW84,moehringSU99,MegowV14}, online scheduling \cite{pruhsST04,FiatW1996-book}, a generalization of the former two \cite{MegowUV06,chouQS06} and robust scheduling \cite{DemeulemeesterH10,kouvelisY97,KasperskiZ2019}. These models differ in the way that information is made available to an algorithm and in the performance metrics. We do not aim at a comprehensive review and, instead, refer the reader to the pointers in the literature. Regarding the access to information, our scheduling with testing model is closest to online and robust optimization, where information (e.g. about job processing times) is revealed incrementally and adversarially. However, in stochastic scheduling, a job's processing time can be explored by partially executing a job and observing its processing time. Clearly, there is much less flexibility in exploiting the learned information, than when testing, as the job might have finished before any action can be taken. 

A new learning-based scheduling model was proposed by Marban, Rutten and Vredeveld~\cite{MarbanRV11}. They introduce a Bayesian model, in which jobs belong to classes and the stochastic processing times of jobs in the same class are drawn from the same unknown distribution. This distribution can be learnt by executing jobs. Besides this Bayesian model and the aforementioned stochastic model of scheduling with testing by Levi et al.~\cite{LeviMS19}, none of the traditional uncertainty models for scheduling takes the opportunity of actively exploring unknown information at some cost into account explicitly.

Finally, it appears noteworthy that the concept of taking exploration cost into account when dealing with uncertainty gains momentum also in other fields such as, e.g., random graphs. Recently some research papers ask the question of how many edges must be queried in a given random graph, in order to verify that some graph property is satisfied. There are results on finding Hamiltonian cycles~\cite{FerberKSV16} and finding paths~\cite{FerberKSV17} in random graphs with few queries.

%!TEX root = main.tex
\section{Preliminaries}\label{sec:preliminaries}
\paragraph{Problem definition.}
The problem of scheduling with testing is defined
as follows.
We are given $n$ jobs to be scheduled on a single machine.
Each job $j$ has an upper limit on the processing time\footnote{We define the problem with rational numbers for the ease of representing them in a computer, but all our results and proofs also hold for real numbers.}~$\p_j\in \mathbbm{Q}^+$. It can either be
executed untested (taking time $\p_j$), or be tested (taking
time~$1$) and then executed at an arbitrary later time
(taking time $p_j\in \mathbbm{Q}^+$, where $0\le p_j\le \p_j$). Initially only
$\p_j$ is known for each job, and $p_j$ is only revealed after
$j$ is tested. The machine can either test or execute a
job at any time. The completion time of job $j$ is denoted by
$C_j$. Unless noted otherwise, we consider the objective of
minimizing the sum of completion times $\sum_j C_j$.

\paragraph{The optimal offline solution.} If the processing times $p_j$ that jobs have after testing are known,
an optimal schedule is easy to determine: Testing and executing
job $j$ takes time $1+p_j$, so it is beneficial to test the job only
if $1+p_j<\bar p_j$. Since the \algoname{SPT} rule is optimal for minimizing the sum of completion times, in the optimal schedule, jobs are ordered by non-decreasing $\min \{1 + p_j, \bar p_j \}$. In this order, the jobs
with $1+p_j<\bar p_j$ are tested and executed while jobs with $1+p_j\ge\bar p_j$
are executed untested. (For jobs with $1+p_j=\bar p_j$ it does not matter
how they are processed.)

\paragraph{Performance analysis.}
We compare the performance of an algorithm $\ALG$ to the optimal
schedule using competitive analysis~\cite{BorodinEY98}. We
denote by $\ALG(I)$ the objective value (cost) of the schedule
produced by $\ALG$ for an instance~$I$, and by $\OPT(I)$ the
optimal cost.
%Note that the optimal schedule tests a job $j$
%if $p_j<\p_j-1$ and executes it untested otherwise. (For
%$p_j=\p_j-1$, both choices are possible.) The time spent
%on job $j$ in the optimal schedule is therefore $\min\{1+p_j,\p_j\}$.
An algorithm $\ALG$ is
\emph{$\mratio$-competitive} or \emph{has competitive
ratio at most $\mratio$} if $\ALG(I)/\OPT(I)\le \mratio$ for all instances
$I$ of the problem. For randomized algorithms, $\ALG(I)$ is
replaced by $E[\ALG(I)]$ in this definition. If the instance $I$ is clear from the context and no confusion can arise, we also write $\ALG$ for $\ALG(I)$ and $\OPT$ for $\OPT(I)$.

% To simplify the analysis, we often consider the limit
% of $A(I)/\OPT(I)$ when the number $n$ of jobs approaches
% infinity.
% Let $\mathcal{I}_n$ denote the set of all instances with $n$ jobs
% with $\p_j\ge 1$ for all jobs. (Jobs with $\p_j<1$ can be optimally \julie{This holds for $\ub_j \leq 1$.}
% executed untested at the start of the schedule in order of non-decreasing
% $\p_j$ and can thus be ignored in the analysis, see also Proposition~\ref{prop:all_jobs_greater_c} below.)
% Note that $\OPT(I)\ge n(n+1)/2$ for $I\in \mathcal{I}_n$ as
% the optimal schedule must spend time $\min\{1+p_j,\p_j\}\ge 1$
% on each job.
% We say that an algorithm $A$ is \emph{asymptotically $\ratio$-competitive}
% \tom{Definition of asymptotic ratio does not fully
% match how we use it.}or
% \emph{has asymptotic competitive ratio at most $\ratio$} if
% %$A(I)\le \ratio \OPT(I)+c$ for a constant $c$ that is
% %independent of the instance.
% $\limsup_{n\to\infty} \sup_{I\in \mathcal{I}_n} A(I)/\OPT(I) \le \ratio$.

When we analyze an algorithm or the optimal schedule, we will typically
first argue that the schedule has a certain structure with
different blocks of tests or job completions. Once we have
established that structure, the cost of the
schedule can be calculated by adding the cost
for each block taken in isolation, plus the effect of the block on
the completion times of later jobs. For example, assume that we have $n$ jobs
with upper limit $\p$, that $\alpha n$ of these jobs are \emph{short}, with
processing time~$0$, and $(1-\alpha)n$ jobs are \emph{long}, with processing time~$\p$.
If an algorithm (in the worst case) first tests the $(1-\alpha)n$ long
jobs, then tests the $\alpha n$ short jobs and executes them immediately,
and finally executes the $(1-\alpha)n$ long jobs that were tested earlier
(see also Figure~\ref{fig:cost}),
the total cost of the schedule can be calculated as
$$
(1-\alpha)n^2 + \frac{\alpha n(\alpha n+1)}{2} +
\alpha n (1-\alpha)n + \frac{(1-\alpha)n((1-\alpha)n+1)}{2}\p
$$
where $(1-\alpha)n^2$ is the total delay that the $(1-\alpha )n$ tests
of long jobs add to the completion times of all $n$ jobs,
$\frac{\alpha n(\alpha n+1)}{2}$ is the sum of completion times
of a block with $\alpha n$ short jobs that are tested and executed,
$\alpha n (1-\alpha)n$ is the total delay that the block of short jobs with
total length $\alpha n$ adds to the completion times of the
$(1-\alpha)n$ jobs that come after it, and
$\frac{(1-\alpha)n((1-\alpha)n+1)}{2}\p$ is the sum of completion
times for a block with $(1-\alpha)n$ job executions with processing
time $\p$ per job.
\begin{figure}[tbh]
	\begin{center}
    \includegraphics[width=10cm]{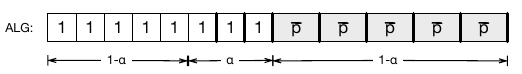}
    \caption{Typical schedule produced by an algorithm. White boxes represent tests % with execution time $1$,
    and grey boxes represent job executions.  The completion time of a job is depicted by a thick bar. Test and execution of a job might be separated. A job of length $0$ completes immediately after its test.}
    \label{fig:cost}
    \end{center}
\end{figure}
\paragraph{Lower limits.}
A natural generalization of the problem would be to allow each
job $j$ to have, in addition to its upper limit $\p_j$, also
a lower limit $\ell_j$, such that the processing time after
testing satisfies $\ell_j\le p_j\le \p_j$. We observe that
the presence of lower limits has no effect on the optimal
schedule, and can only help an algorithm. As we are interested
in worst-case analysis, we assume in the remainder of the paper
that every job has a lower limit of~$0$. Any algorithm that
is $\rho$-competitive in this case is also $\rho$-competitive
in the case with arbitrary lower limits (the algorithm can simply
ignore the lower limits).

\paragraph{Preemption.} The ability to preempt the execution of a test or of a (tested or untested) job would be of no benefit to the algorithm or the adversary
as no new information is obtained during the execution. Therefore, we only consider algorithms and schedules that do not preempt tests and that do not preempt job executions. However, as noted above, the execution of a tested job can be scheduled any time after the completion of the test.

\paragraph{Jobs with small $\bar p_j$.}
We will consider several algorithms and prove competitiveness for them. We observe that any $\rho$-competitive algorithm may process jobs with $\bar p_j < \rho$ without testing in order of increasing $\bar p_j$ at the beginning of its schedule.
\begin{lemma}\label{lem:all_jobs_greater_c}
   Without loss of generality any algorithm $\ALG$ (deterministic or randomized) claiming competitive ratio $\rho$ starts by scheduling untested all jobs~$j$ with $\ub_j<\rho$ in increasing order of $\ub_j$. Moreover, worst case instances for $\ALG$ consist solely of jobs~$j$ with $\ub_j \geq \rho$.
\end{lemma}
% \LEMalljobsgreaterc*
%
\begin{proof}
    We transform $\ALG$ into an algorithm $\ALG'$ which obeys the claimed behavior and show that its ratio does not exceed $\rho$.
    Consider an arbitrary instance $I$.

    Let $J$ be the sequence of jobs $j$ with $\ub_j < \rho$ ordered in increasing $\ub_j$ order.  We divide $J$ into $J_0, J_1$, where $J_0$ consists of the jobs $j$ with $0\leq \ub_j<1$ and $J_1$ consists of the jobs $j$ with $1\leq \ub_j < \rho$.
    $\ALG'$ starts by executing the job sequence $J$ untested, and then schedules all remaining jobs as $\ALG$, following the same decisions to test and the order of tests and executions.
    In a worst-case instance all the jobs in $J$ have processing time $0$.  By optimality of the \algoname{SPT} policy $\OPT$ schedules first $J_0$ untested as well, and then schedules $J_1$ tested spending time $1$ on each job.
    The ratio of the costs of these parts is 
    \[
        \frac{\ALG'(J)}{\OPT(J)} < \rho
    \]
    where the inequality follows from $\ub_j/\min\{1,\ub_j\} < \rho$ for all $j\in J$.  Let $\textrm{len}$ denote the length of a schedule. Then by the same argument we have
    \[
        \frac{\textrm{len}(\ALG'(J))}{\textrm{len}(\OPT(J))} < \rho.
    \]

    Let $I'$ be the instance $I$ without the jobs in $J$.  Let $k$ be the number of jobs in $I'$. %\chris{The proof works with $k=0$. Hence I removed: We assume $k>0$ otherwise we are done with the proof.}
    Since $I'$ contains only jobs with large upper limit, we have $\ALG(I')=\ALG'(I')$.
    We have
    \begin{align*}
        \ALG'(I) &= \ALG'(J) + k\cdot\textrm{len}(\ALG'(J)) + \ALG'(I')
    \\
        \OPT(I) &= \OPT(J) + k\cdot\textrm{len}(\OPT(J)) + \OPT(I').
    \end{align*}
    From these (in)equalities we conclude
    \begin{align*}
        \frac{\ALG(I)}{\OPT(I)} \leq \rho &\Rightarrow
        \frac{\ALG'(I)}{\OPT(I)} \leq \rho
        \\
        \frac{\ALG(I)}{\OPT(I)} \geq \rho &\Rightarrow
        \frac{\ALG'(I)}{\OPT(I)} \le \frac{\ALG(I')}{\OPT(I')}
    \end{align*}
    which means that if $\ALG$ is $\rho$ competitive then so is $\ALG'$ and that there are worst-case instances for $\ALG$ only with jobs having upper limit at least $\rho$.
    \qed
\end{proof}

\paragraph{Increasing or decreasing \ALG and \OPT.}

Throughout the paper we sometimes consider worst-case instances consisting of only a few different job types.  In order to do so we need to change carefully the parameters of a given instance, in such a way that the competitive ratio does not decrease and the number of distinct job types decreases.  The following generic proposition allows us to do so in some cases.

\begin{proposition}\label{prop:change-indep-deltap}
    Fix some algorithm $\ALG$ and consider a family of instances described by some parameter $x\in[\ell,u]$, which could represent $p_j$ or $\ub_j$ for some job $j$ or for some set of jobs. Suppose that both $\OPT$  and $\ALG$ are linear in $x$ for the range $[\ell,u]$.  Then the ratio $\ALG/\OPT$ is maximized, among all choices of $x\in [\ell,u]$, for at least one of the two choices $x=\ell$ or $x=u$. Moreover, if $\OPT$ and $\ALG$ are increasing in $x$ with the same slope, then this holds for $x=\ell$.
\end{proposition}
% \PROPchangeindepdeltap*
\begin{proof}
    The proof follows from the fact that an expression of the form $\ALG/\OPT=(a+b x)/(a'+b' x)$ is monotone in $x$.  Indeed its derivative is
    \[
            \frac{a' b - a b'}{(a' + b' x)^2}
    \]
    whose sign does not depend on $x$.
    The last statement follows from the fact that if $\ALG>\OPT$ and $0<\delta\leq \OPT$ then $(\ALG-\delta)/(\OPT-\delta) > \ALG/\OPT$.
    \qed
\end{proof}
We can make successive use of this proposition in order to show useful properties on worst-case instances.
\begin{lemma}\label{lem:linear-pj}
	% \new{Fix some algorithm $\ALG$ and
 %   suppose that there is an interval $[\ell,u]$ such that
 %   the algorithm processes all jobs $j$ with $p_j\in[\ell,u]$ consecutively and in the same manner, meaning either all preceded by a test or none.
 %   Suppose that this is also the case for OPT.  Then there is worst case instance where every job $j$ with $p_j\in[\ell,u]$ satisfies $p_j\in\{\ell,u\}$.
 %   }
    Suppose that there is an interval $[\ell',u']$ such that \OPT schedules all jobs~$j$ with $p_j\in[\ell',u']$ either all tested or all untested, independently of the actual processing time in $[\ell',u']$.  Suppose that this holds also for $\ALG$.
    %Moreover suppose that both $\OPT$ and $\ALG$ are \tom{Is ``insensitive'' clear enough?}
    %\nic{Moreover suppose that the schedules of $\OPT$ and $\ALG$ do not change when changing the  processing times in $[\ell',u']$ as long as the relative ordering of job processing times does not change.}
    %insensitive to changes of the processing times in $[\ell',u']$ which maintain the ordering of processing times.
Moreover, suppose that the schedules of $\OPT$ and $\ALG$ do not change
(in the sense that the order of all tests and job executions remains the same)
when changing the processing times in $[\ell',u']$ as long as the
relative ordering of job processing times does not change.
    Then there is a worst-case instance for $\ALG$ where every job $j$ with $p_j\in[\ell',u']$ satisfies $p_j\in\{\ell',u'\}$.
\end{lemma}
% \CORlinearpj*
\begin{proof}
    Fix some worst-case instance for the algorithm $\ALG$.
    Let $S$ be the set of jobs $j$ with $p_j = x$ for some $x$ with $\ell'< x < u'$.  
    Let $\ell,u$ be the values  $\ell=\max(\{\ell'\}\cup\{p_i:p_i<x\})$ and
    $u=\min(\{u'\}\cup\{p_i:p_i>x\})$.
    Informally $\ell$ is the largest processing time strictly smaller than $x$ or $\ell'$ if $x$ is already the smallest processing time or if this would make $\ell$ smaller than $\ell'$. 
    Also $u$ is the smallest processing time strictly larger than $x$ or $u'$ if $x$ is already the largest processing time or if this would exceed $u'$. 
    Since the schedules are preserved when changing the processing times of $S$, both costs $\ALG$ and $\OPT$ are linear in $x$ within $[\ell,u]$.
    Now we can use Proposition~\ref{prop:change-indep-deltap} to show that there is a worst-case instance where all jobs in $S$ have processing time either $\ell$ or $u$.  In both cases we have reduced the number of distinct processing times strictly being between $\ell'$ and~$u'$.  By repeating this argument sufficiently often we obtain the claimed statement.
    \qed
\end{proof}

\section{Deterministic Algorithms}\label{sec:det}
%!TEX root = main.tex

\subsection{Algorithm \BOUND} % (CR $2$)}
\label{sec:bound2}%

%We show a competitive ratio of $2$ for a natural threshold algorithm with parameter $T=2$.
We show a competitive ratio of $2$ for a natural algorithm that uses a threshold
to decide whether to test a job or execute it untested.

\begin{algo}[\BOUND]
 First jobs with $\p_j< 2$ are scheduled in order of non-decreasing upper limits without testing. Then all remaining jobs are
tested. If the revealed processing time of job~$j$ is $p_j \leq 2$ (short jobs), then the job
is executed immediately after its test. After all pending jobs (long jobs) have been
tested, they are scheduled in order of increasing processing time $p_j$.
\end{algo}

By Lemma~\ref{lem:all_jobs_greater_c} we may restrict our competitive analysis w.l.o.g.\ to instances  with $\p_j \geq 2$. Note, that on such instances \BOUND tests all jobs. From a simple interchange argument it follows that the structure of the algorithm's solution in a worst-case instance is as follows.
\smallskip
\begin{compactitem}
\item Test phase: The algorithm tests all jobs that have $p_j>2$, and defers them.
\item Short jobs phase: The algorithm tests short jobs ($p_j\le 2$) and executes each
  of them right away. The jobs are tested in order of non-increasing processing time.
\item Long jobs phase: The algorithm executes all deferred long jobs in order of
  non-decreasing processing times.
\end{compactitem}
\medskip

%An optimal solution will not test jobs with $p_j+1 \geq \p_j$. If $\p_j\in (2,3]$, then such an untested job is a small job, otherwise it is among the long jobs. The optimal schedule sorts jobs in non-decreasing order of values $\min\{1+p_j, \p_j\}$.
An optimal solution will not test jobs with $p_j+1 \geq \p_j$.
It sorts jobs in non-decreasing order of values $\min\{1+p_j, \p_j\}$.

First, we analyze and simplify worst-case instances.
% \new{
% 	First we observe that in a worst case instance we have $\ub_j=\{2,p_j\}$. Reducing the upper limit does not affect the algorithm and can only help the optimal solution.

% 	Let $\epsilon>0$ be an arbitrary small number such that all jobs $j$ with $p_j>2$ satisfy $p_j>2+\epsilon$. Let $\ell=2+\epsilon$ and $u=\max p_j$.  Then both \ALG and \OPT schedule these jobs untested, without interruption, and the end of the schedule and in the same order non-decreasing in $p_j$.  Hence Corollary~\ref{cor:linear-pj} implies that in a worst case instance all these jobs have all processing time $x$ being $\ell$ or $u$.  However since these jobs are identically scheduled and at the end, both $\ALG$ and $\OPT$ are linear in $x$ and with the same slope. Therefore in a worst case instance these jobs have all processing time $2+\epsilon$.

% 	Now we focus on jobs with processing time between $\ell=1$ and $u=2$.  They are all tested and immediately executed by $\ALG$ in order of non-increasing processing time, while they are scheduled untested by $\OPT$ in the opposite order.  Corollary~\ref{cor:linear-pj} implies that in a worst case instance these processing times are either $1$ or $2$.

% 	Finally we focus on jobs with process time between $\ell=0$ and $u=1$.
% }
% \LEMwcshortjobs*
\begin{lemma}\label{lem:wc-short-jobs}
	There is a worst-case instance for \BOUND in which all short jobs with $p_j\leq 2$ have processing time either $0$ or $2$. %This is true also for instances with fixed uniform upper limit~$\p_j=\p$.
\end{lemma}
We give a proof without modifying upper limits, which is not necessary in this section but will come handy later when we analyze \BOUND for arbitrary uniform upper limits.
\begin{proof}
	Consider short jobs that are tested by both, the optimum and \BOUND, i.e., short jobs with $p_j<\p_j-1$. We argue that we can either decrease the processing time of a short job $j$ to $0$ or increase it to $\min\{2, \p_j-1\}$ without decreasing the worst-case ratio. With respect to the order in which \BOUND executes the jobs, let $\ell$ be the first short job with $p_{\ell} < \min\{2, \p_{\ell} - 1\}$ and let $i$ be the last short job with $p_i > 0$.

	Suppose $i\neq \ell$. Let $\Delta=\min\{p_i, \min\{2, \p_{\ell} - 1\}  - p_{\ell}\}$. We decrease $p_i$ by $\Delta$ and at the same time increase $p_{\ell}$ by $\Delta$. The value $\Delta$ is chosen in such a way that either $p_i$ will become $0$ or $p_{\ell}$ will be $\min\{2, \p_{\ell} - 1\}$, as desired. The schedule produced by the algorithm will be the same except that jobs $\ell,\ldots,i-1$ complete $\Delta$ units later. In the optimal schedule $\ell$ and $i$ are scheduled in opposite order. Suppose we keep the schedule fixed when changing the processing times of jobs $i$ and $\ell$. Then $i$'s completion time as well as those of jobs between $i$ and $\ell$ decreases. In an optimal schedule jobs might be re-ordered, but this only improves the total objective further. Hence, the total ratio of objective values does not decrease.

	Now, assume $i = \ell$, i.e., there is exactly one short job with processing time $p_i$ strictly between~$0$ and $\min\{2,\p_i-1\}$. We argue that either increasing or decreasing~$p_i$ to $\min\{2,\p_i-1\}$ or $0$ will not decrease the worst-case ratio. Such a change~$\Delta$ does not change the order of jobs in the algorithm's solution and thus the change in the objective is~$\Delta$ times the number of jobs completing after $i$. In an optimum solution, there are untested short or long jobs which are scheduled between short tested jobs and their relative order with $i$ may change when $i$ is in-/decreased by $\Delta$. However, let us consider a possibly not optimal schedule that simply does not adjust the order after changing $i$. Then the change in the objective is linear in $\Delta$ in the above-given range, as it is for the algorithm, and thus, by Proposition~\ref{prop:change-indep-deltap} either increasing or decreasing $p_i$ by $\Delta$ does not decrease the ratio of objective values. Now, the truly optimal objective value is not larger and thus, the true worst-case ratio is not smaller.

	Now, we may assume that all short jobs remaining with processing times different from $0$ and $2$ are untested in the optimum solution because their processing time is at least $\p_j-1$. Again, the optimum does not test those jobs, and hence, increasing the processing time to $2$ has no impact on the optimal schedule, while our algorithm's cost only increases. Thus, the worst-case ratio increases, which concludes the proof.
	\qed
	\end{proof}

\BOUND tests all jobs and makes scheduling decisions depending on job processing times~$p_j$ but independently of upper limits of jobs. Since all short jobs have $p_j\in \{0,2\}$, we can reduce all their upper limits to $\p_j=2$ without affecting the schedule, whereas it may only improve the optimal schedule. In particular we may assume now the following.

\begin{lemma}\label{prop:short}
	There is a worst-case instance in which all short jobs have $\p_j=2$ and execution times are $0$ or $2$.
\end{lemma}

% \LEMwclongjobs*
\begin{lemma}\label{lem:wc-long-jobs}
	There is a worst-case instance in which long jobs with $p_j>2$ have a uniform upper limit $\p$ and processing times $p_j=\p_j=2+\eps$ for infinitesimally small $\eps>0$.
\end{lemma}
\begin{proof}
	For all long jobs, which are tested by the optimum, we reduce the upper limit to $\p_j=1+p_j$. This does not change  the algorithm's solution. But the optimum may as well run those previously tested jobs also untested and would not change its  total objective value.

	Now the optimum solution runs all long jobs without testing them. Thus, increasing the processing time of long jobs to $p_j=\p_j$ does not affect the optimum cost whereas the algorithm's cost increases.

	Lemma~\ref{prop:short} implies that all long jobs are scheduled in the same order by the algorithm and an optimum without any short jobs in between. Then, setting $\p=2+\eps$ decreases the objective values of both algorithms by the same amount and thus does not decrease the ratio. The lemma follows.
	\qed
	\end{proof}

Now we are ready to prove the main result.
% \THMbound*
\begin{theorem}\label{th:bound2}%
 Algorithm \BOUND has competitive ratio at most $2$ for scheduling with testing %on a single machine
with the objective of minimizing the sum of completion times.
\end{theorem}
\begin{proof}
	We consider worst-case instances of the type derived above. Let $a$ be the number of short jobs with $p_j=0$, let $b$ be the number of short jobs with $\p_j=p_j=2$, and let~$c$ be the number of long jobs with $\p_j=2+\eps$, see Figure~\ref{fig:threshold:ap}.

	\begin{figure}
	\begin{center}
		\includegraphics[width=10cm]{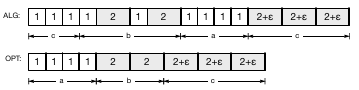}
		\caption{Worst case instance for \BOUND.}
		\label{fig:threshold:ap}
	\end{center}
	\end{figure}

 	\BOUND's solution for a worst-case instance first tests all long jobs, then tests and executes the short jobs in decreasing order of processing times, and completes with the executions of long jobs. The total objective value $ALG$ is
 	\begin{multline*}
		\ALG = (a+b+c)c + b(b+1)/2\cdot 3 +3b(a+c) + \\
			+ a(a+1)/2 + a\cdot c + c(c+1)/2\cdot (2+\eps) .
	\end{multline*}

	An optimum solution tests and schedules first all $0$-length jobs and then executes the remaining jobs without tests. The objective value is
	\begin{align*}
		\OPT &= a(a+1)/2 + a(b+c) +b(b+1)/2\cdot 2 + 2bc + c(c+1)/2\cdot (2+\eps).
	\end{align*}
	Simple transformation shows that $\ALG\leq 2\cdot \OPT$ is equivalent to
		\begin{align*}
			2 ab + 2c^2 & \leq a^2 + b^2 + a + b + c(c+1)(2+\eps) &
			\Leftrightarrow
			\\
			0 &\leq (a-b)^2 + a + b + c^2 \eps + c(2+\eps),
		\end{align*}
	which is obviously satisfied and the theorem follows.
	\qed
\end{proof}

%\begin{remark}\nic{remove or make more precise; maybe better some statement in conclusions. I am not convinced that we safely checked this.}
%    If we parameterize the analysis of \BOUND to show approximation ratio $T$, the bound for $\p_j$ for which we start to test jobs, we cannot decrease $T$ below $2$. Also, randomizing the choice of the next job to test, does not allow improvement.
%\end{remark}

Note that the analysis of  \BOUND is tight. Indeed it has ratio $2-\varepsilon$ on the instance consisting of a single job $j$ with $p_j=0$ and $\p_j=2-\varepsilon$, for arbitrarily small $\varepsilon > 0$.  The algorithm does not test the job, but the optimal schedule does.

%We conclude this section with an observation that was brought to our attention. Consider a slight modification of \BOUND, which delays all jobs after their test, regardless of the revealed processing time.  This algorithm seems to produce worse schedules than \BOUND. For example when all jobs have upper bound 2 and processing time 0, this algorithm has a cost of $n^2$, while \BOUND has a cost of $n(n+1)/2$. However the competitive ratio of this simple algorithm is also 2, as can be shown with a simple argument.  Again, by Lemma~\ref{lem:all_jobs_greater_c} we might assume that all jobs have upper limit at most 2, and hence the algorithms starts testing all jobs, then executes them in order of non-increasing processing times.  The resulting cost is $n^2+A$, where $A$ is the cost of the second part of the schedule. This however it at most $\OPT$, since every job in the second phase has processing time $p_j$ while in the optimal schedule it needs time $\min\{1+p_j, \p_j\} \geq p_j$ to be processed.  Since these times are at least 1, we have $\OPT\geq n(n+1)/2$.  In conclusion we have $n^2+A \leq 2 \cdot \OPT$.

%Note that the analysis of \BOUND is tight. Indeed, consider instances of
%the form shown in Figure~\ref{fig:threshold:ap} with $a=b=c=\frac{n}{3}$.
%In the limit for $\varepsilon\rightarrow0$, we get $\OPT=\frac{13}{18}n^2+O(n)$
%and $\ALG=\frac{26}{18}n^2+O(n)$, showing that the competitive ratio
%approaches~$2$ for~$n\rightarrow\infty$.

We conclude
this section with an observation that was brought to our attention. Consider a
slight modification of \BOUND that delays \emph{all} jobs after their test,
regardless of the revealed processing time. This algorithm, which we
call \TA, seems to produce
worse schedules than \BOUND. For example, when all jobs have upper bound~$2$ and
processing time~$0$, \TA has a cost of $n^2$, while \BOUND has a cost
of $n(n+1)/2$. Nevertheless, the competitive ratio of \TA is
also~$2$, which can be shown as follows: Again,
by Lemma~\ref{lem:all_jobs_greater_c} we may assume that all jobs have upper
limit at least~$2$. Hence, \TA starts by testing all jobs, and then
executes them in order of non-decreasing processing times. By Lemma~\ref{lem:linear-pj},
we can assume that all jobs $j$ with $0\le p_j\le 1$ (which are tested by
both \OPT and \TA) satisfy $p_j\in\{0,1\}$.
For the jobs with $p_j\in [1,2]$, we can then first set $\bar{p}_j=2$ (this can only
help \OPT) and next set $p_j=2$ (this can only increase the cost of \TA but does
not affect \OPT).
Finally, we can set $\bar{p}_j=p_j$ for all jobs with $p_j>2$ (this can only
help \OPT) and then set $p_j=\bar{p}_j=2$
for all these jobs (this decreases the objective values of \OPT
and \TA by the same amount and thus does not decrease the ratio).
Let the resulting instance consist of $a$ jobs with processing time~$0$ and
$b$ jobs with processing time~$2$. The cost of \TA
is $A=(a+b)(a+b)+b(b+1)=a^2+2ab+2b^2+b$, and the cost of $\OPT$ is
$\frac12 a(a+1)+ab+b(b+1)=\frac{a^2}{2}+\frac{a}{2}+ab+b^2+b\ge\frac12 A$.

%!TEX root = main.tex

\subsection{Deterministic lower bound}\label{subsec:detLB}
In this section we give a lower bound on the competitive ratio of any deterministic algorithm. The instances constructed by the adversary have a very special form: All jobs have the same upper limit $\bar{p}$, and the processing time of every job is either $0$ or $\bar{p}$.

Consider instances of $n$ jobs with uniform upper limit $\bar p>1$, and consider any deterministic algorithm. We say that the algorithm \emph{touches} a job when it either tests the job or executes it untested. We re-index jobs in the order in which they are first touched by the algorithm, i.e., job $1$ is the first job touched by the algorithm and job $n$ is the last. The adversary fixes a fraction $\delta\in[0,1]$ and sets the processing time of job $j$, $1\le j\le n$, to:
$$
p_j = \begin{cases}
    0               &       \text{, if $j$ is executed by the algorithm untested, or $j> \delta n$}    \\
    \bar p      &       \text{, if $j$ is tested by the algorithm and $j\leq \delta n$}
\end{cases}$$
A job $j$ is called \emph{short} if $p_j=0$ and \emph{long} if $p_j=\bar p$.
Let $j_0$ be the smallest integer that is greater than~$\delta n$.
Job $j_0$ is the first of the last $(1-\delta)n$ jobs that are short no matter
whether the algorithm tests them or not.

%It will turn out that the best choices of $\delta$ and $\bar p$ for the adversary
%are $\delta\approx 0.6306655$ and $\bar p\approx 1.98962$, leading to a lower bound
%of~$1.854628$ on the competitive ratio of any deterministic algorithm.

We assume the algorithm knows $\bar p$ and $\delta$, which can only improve the performance of the best-possible deterministic algorithm.
Note that with $\delta$ and $\bar p$ known to the algorithm, it has full information about the actions of the adversary.
Nevertheless, it is still non-trivial for an algorithm to decide for each of the first $\delta n$ jobs whether to test it
(which makes the job a long job, and hence the algorithm spends time $\bar p+1$ on it while the optimum executes
it untested and spends only time $\bar p$) or to execute it
untested (which makes it a short job, and hence the algorithm spends time $\bar{p}$ on it while the optimum
spends only time~$1$).

\begin{figure}
\begin{center}
  \includegraphics[width=12cm]{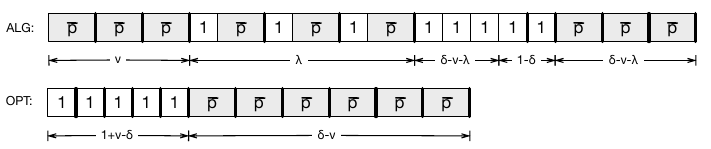}
  \caption{Lower bound construction.}
  \label{fig:detLB2}
  \end{center}
\end{figure}

Let us first determine the structure of the schedule produced by an algorithm that achieves the best possible competitive ratio for instances created by this
adversary, as displayed in Figure~\ref{fig:detLB2}.

%\begin{observation}
%    A dominant deterministic algorithm for this adversarial sequence has the following structure: It executes all $u n$ untested jobs at the beginning of the schedule, then tests $\Gamma n$ jobs and executes them immediately, then tests $n(\Delta - \Gamma - u)$ jobs and delays their execution, then tests and executes the remaining $n(1 - \Delta)$ jobs of length $0$, and finally executes the delayed jobs.
%\end{observation}

% \LEMstructureDetLB*
\begin{lemma}
The schedule of a deterministic algorithm with best possible competitive ratio
has the following form, where $\lambda,\nu\ge 0$ and $\nu+\lambda\le\delta$:
The algorithm first executes $\nu n$ jobs untested,
then tests and executes $\lambda n$ long jobs, then tests $(\delta-\nu-\lambda)n$
long jobs and delays their execution, then tests and executes the remaining
$(1-\delta)n$ short jobs, and finally
executes the $(\delta-\nu-\lambda)n$ delayed long jobs that were tested earlier, see Figure~\ref{fig:detLB2}.
\end{lemma}

\begin{proof}
It is clear that the algorithm will test
the last $(1-\delta)n$ jobs and execute each such job (with processing time~$0$)
right after its test, as executing any of them untested does not affect the optimal solution but increases the objective value of the algorithm. Furthermore, consider the time $t$ when the algorithm
tests job~$j_0$. From this time until the end of the schedule, the algorithm
will test and execute the last $(1-\delta)n$ jobs (spending time $1$ on each such job),
and execute all the long jobs that were tested earlier but not yet
executed (spending time $\bar p>1$ on
each such job). As the \algoname{SPT} rule is optimal for minimizing the sum of completion times,
it is clear that from time $t$ onward the algorithm will first test and execute
the $(1-\delta)n$ short jobs and afterwards execute the long jobs that were tested but
not executed before time~$t$.

Before time~$t$, the algorithm touches the first $\delta n$ jobs. Each of these
can be executed untested (let $\nu n$ be the number of such jobs), or tested and also
executed before time~$t$ (let $\lambda n$ be the number of such jobs), or tested
but not executed before time~$t$ (this happens for the remaining $(\delta-\nu-\lambda)n$ jobs).
To minimize the sum of completion times of these jobs, it is clear that the
algorithm first executes the $\nu n$ jobs untested (spending time $\bar p$ per job),
then tests the $\lambda n$ long jobs and executes each of them right after its
test (spending time $1+\bar p$ per job), and finally tests the
remaining $(\delta-\nu-\lambda)n$ long jobs.
\qed
\end{proof}

The cost of the algorithm in dependence
on $\nu$, $\lambda$, $\delta$ and $\bar p$ can now be expressed as:
\begin{eqnarray*}
\lefteqn{\ALG(\nu,\lambda,\delta,\bar p) =
n^2\Big(
\frac{\nu^2}{2}\bar p + \nu \bar p (1-\nu)
+ \frac{\lambda^2}{2} (1+\bar p) + \lambda (1+\bar p) (1-\nu-\lambda)}\\
&& \mbox{}+ (\delta-\nu-\lambda)(1-\nu-\lambda)
+ \frac{(1-\delta)^2}{2} + (1-\delta) (\delta-\nu-\lambda)
+ \frac{(\delta-\nu-\lambda)^2}{2} \bar p\Big)+O(n)\\
&=& \frac{n^2}{2}\left(1 + 2\delta(1-\nu \bar p) + \delta^2(\bar p-1)+ 2\nu(\nu+\bar p-2)
+ \lambda^2 +2\lambda(\nu+\bar p-1-\delta\bar p)\right)+O(n)
\end{eqnarray*}
The optimal schedule first tests and executes the $(\nu+1-\delta)n$ short
jobs and then executes the $(\delta-\nu)n$ long jobs untested. Hence,
the optimal cost, which depends only on $\nu$, $\delta$ and
$\bar p$, is:
\begin{eqnarray*}
\OPT(\nu,\delta,\bar p) &=&
n^2\left(\frac{(\nu+1-\delta)^2}{2} + (\nu+1-\delta)(\delta-\nu) + \frac{(\delta-\nu)^2}{2}\bar p\right)+O(n)\\
& =&  \frac{n^2}{2}\left(1+(\delta-\nu)^2(\bar p-1)\right)+O(n)
\end{eqnarray*}
We introduce the notations
\begin{align*}
\ALG'(\nu,\lambda,\delta,\bar p) &= \lim_{n\to\infty} \frac{2}{n^2}\ALG(\nu,\lambda,\delta,\bar p)
&\mbox{and}
\\
\OPT'(\nu,\delta,\bar p) & = \lim_{n\to\infty} \frac{2}{n^2}\OPT(\nu,\delta,\bar p).
\end{align*}
As the
adversary can choose $\delta$ and $\bar p$, while the algorithm can choose $\nu$ and $\lambda$,
the value
$$
R = \max_{\delta,\bar p} \min_{\nu,\lambda} \frac{\ALG'(\nu,\lambda,\delta,\bar p)}{\OPT'(\nu,\delta,\bar p)}
$$
gives a lower bound on the competitive ratio of any deterministic algorithm
in the limit for $n\to\infty$.
By making $n$ sufficiently large, the adversary can create instances with
finite $n$ that give a lower bound arbitrarily close to $R$.

The exact optimization of $\delta$ and $\bar p$ is rather tedious and
technical as it involves the optimization of rational functions of
several variables. In the following, we therefore only show that the choices
$\delta= 0.6306655$ and $\bar p= 1.9896202$ give a lower bound
of~$1.854628$ on the competitive ratio of any deterministic algorithm.
(The fully optimized value of $R$ is less than $1.1\cdot10^{-7}$ larger
than this value.) For this choice of $\delta$ and $\bar p$ we have:
\begin{eqnarray*}
\ALG'(\nu,\lambda,\delta,\bar p)
&\approx&  1.32747 + \nu(\nu-1.26516)+\frac12 \lambda^2 +\lambda(\nu-0.265165)\\
\OPT'(\nu,\delta,\bar p) &\approx& 0.696805+\nu(0.49481\nu-0.624119)
\end{eqnarray*}
The part of $\ALG'(\nu,\lambda,\delta,\bar p)$
involving $\lambda$ is $\frac12\lambda^2+\lambda(\nu+\bar p-1-\delta\bar p)$, which
is a quadratic function minimized at $\lambda=1+\delta\bar p-\bar p-\nu\approx 0.265165-\nu$.
As $\lambda$ must be non-negative,
we distinguish two cases depending on whether this expression is non-negative or not. Let $\tau=1+\delta\bar p-\bar p\approx 0.265165$.

\paragraph{Case 1:} $\nu\le \tau$. In this case the best choice of $\lambda$
for the algorithm is $\lambda=\tau-\nu$. The ratio $\ALG'/\OPT'$ then simplifies to:
$$
f(\nu)=\frac{1.29231+\nu(\frac12\nu-1)}{0.696805+\nu(0.49481\nu-0.624119)}
=
1.01049+\frac{1.18874-0.746417\nu}{1.40823+\nu(\nu-1.26133)}
$$
In the range $0\le\nu\le \tau$,
the only local extremum of this function is a local maximum at
$\nu\approx 0.201266$, so the function attains its minimum
in the range at one of the two endpoints. As we have
$f(\tau)> f(0)\approx 1.854628$,
the function is minimized at $\nu=0$, giving a lower bound of
$1.854628$ on the competitive ratio.

\paragraph{Case 2:} $\nu>\tau$.
In this case, the best choice of $\lambda$ for the algorithm is $\lambda=0$. The
ratio $\ALG'/\OPT'$ then becomes:
$$
g(\nu)=
\frac{1.32747+\nu(\nu-1.26516)}{0.696805+\nu(0.49481\nu-0.624119)}
=
2.02098+\frac{-0.163208-0.00774781\nu}{1.40823+\nu(\nu-1.26133)}.
$$
This function is monotonically decreasing in
the range $\tau<\nu\le\delta$, so it is
minimized for $\nu=\delta$, giving a ratio of $g(\delta)\approx 1.854628$.

As we get a lower bound of $1.854628$ in both cases, this lower
bound holds generally.

% \THMdetLB*
\begin{theorem}\label{thm:detLB}
No deterministic algorithm can achieve a competitive ratio or asymptotic
competitive ratio below $1.854628$ for scheduling with testing %on a single machine
with the objective of minimizing the sum of completion times. This holds even for instances with
uniform upper limit where each processing time is either $0$ or equal to the
upper limit.
\end{theorem}

\section{Randomized Algorithms}\label{sec:rand}
%!TEX root = main.tex

\subsection{Algorithm \RTE}

\begin{algo}[\RTE]
The randomized algorithm \RTE has parameters $1\leq T\leq E$ and works in 3 phases. First it executes all jobs with $\ub_j<T$ without testing in order of increasing $\ub_j$. Then it tests all jobs with $\bar{p}_j\ge T$ in uniform random order. Each tested job $j$ is executed  immediately after its test if $p_j\le E$ and is deferred otherwise. Finally all deferred jobs are executed in order of increasing processing time.
\end{algo}

% The structure of the schedule is as follows:
% \begin{itemize}
% \item first jobs with $\bar{p}_j<T$ are scheduled without testing in increasing order of $\bar{p}_j$
% \item then jobs with $\bar{p}_j\ge T$ are tested in random order, with those of processing time at most $E$ being immediately executed
% \item last the deferred jobs that were tested and have processing time greater than $E$ are executed in order of non-decreasing processing times
% \end{itemize}

% Clearly the competitive ratio of \RTE is at least $T$.  Simply consider an instance consisting of a single job $j$ with $\ub_j=T-\epsilon$ and $p_j=0$ for an arbitrary small $\epsilon>0$.  The ratio is $T-\epsilon$ for this instance.

We analyze the competitive ratio of \RTE, and optimize the parameters $T,E$ such that the resulting competitive ratio is $T$.

By Lemma~\ref{lem:all_jobs_greater_c} we restrict to instances with $\ub_j \geq T$ for all jobs.  Then, the schedule produced by \RTE can be divided into two parts. Part~(1) contains all tests, of which those that yield processing time $p_j$ at most $E$ are immediately followed by the job's execution. Part~(2) contains all jobs that \emph{have been} tested and with processing time larger than $E$. These jobs are ordered by increasing processing time. Jobs in the first part are completed in an arbitrary order.

Furthermore, we can assume $\ub_j = \max \{ p_j, T\}$ for all jobs.  Reducing $\ub_j$ to this value does not change the cost or behavior of \RTE, but may decrease the cost of $\OPT$.  We make further assumptions along the following lines.
Let $\epsilon>0$  be an arbitrary small number such that $p_j\geq E+\epsilon$ for all jobs $j$ with $p_j>E$.  These jobs are executed by \RTE in part (2) of the schedule in non-decreasing order of processing time.  The same holds for $\OPT$, which by the \emph{\algoname{SPT} Policy} also schedules these jobs in the end in exactly the same order.  Hence if we set $\ub_j=p_j=E+\epsilon$ for all these jobs, then we reduce the objective value of \RTE and of $\OPT$ by the same value.  According to Proposition~\ref{prop:change-indep-deltap} this transformation only increases the competitive ratio of the algorithm.

Using again the assumption that $\ub_j = \max \{ p_j, T\}$ for all jobs, we now have that
%\sout{all jobs $j$ satisfy either $\ub_j=p_j=E+\epsilon$ or $\ub_j = p_j \in [T,E]$} 
all jobs $j$ in part~(2) satisfy $\ub_j=p_j=E+\epsilon$ and the remaining jobs satisfy either $\ub_j = p_j \in [T,E]$ or $\ub_j = T$ and $p_j\leq T$.
Now we apply Lemma~\ref{lem:linear-pj} to show that for all jobs $j$ with $\ub_j =  p_j\in[T,E]$ we can in fact assume $\ub_j =  p_j\in\{T,E\}$.
%\textcolor{gray}{The usage of the lemma is a bit subtle as the output of \RTE is a distribution of schedules and as we change the processing times and upper limits of jobs simultaneously.  However for each fixed order the conditions of the statement of the lemma are satisfied.  Then we conclude using that the expected completion time of \RTE is a linear combination of the objective values over each of the $n!$ orders.}
%
The usage of the lemma is a bit subtle as the output of \RTE is a distribution of schedules. For any fixed scheduling order corresponding to a realization of the random execution of the algorithm, the conditions of the lemma are satisfied. But we cannot apply the lemma on each order individually, as we might end up with different problem instances. However, the expected cost of \RTE is linear in the execution times of jobs $j$ within $p_j\in[T,E]$.  This is the key condition which is used in the proof of Lemma~\ref{lem:linear-pj}. Hence we conclude that the statement of the lemma still holds.

Now we turn to jobs $j$ with $\ub_j=T$ and $p_j \leq T$.  For the jobs with $0\leq p_j \leq T-1$, the same argument implies that $p_j\in\{0, T-1\}$.  However jobs $j$ with $\ub_j = T$ and $T-1\leq p_j \leq T$ are not tested in $\OPT$.  Therefore increasing their processing time to $p_j=T$ does not change $\OPT$ but increases the cost of \RTE and therefore increases the competitive ratio.

In conclusion a worst case instance is described completely by the number of jobs $n$ and fractions $\alpha,\beta,\gamma$ as follows, see Figure~\ref{fig:randUB:ap}.
% We denote by $\alpha,\beta,\gamma,1-\alpha-\beta-\gamma$ the fractions of these jobs.
\begin{itemize}
    \item A $1-\alpha-\beta-\gamma$ fraction of the jobs have $\ub_j=T$ and $p_j=0$. (type 0 jobs)
    \item An $\alpha$ fraction of the jobs have $\ub_j=T$ and $p_j=T$. (type T jobs)
    \item A $\beta$ fraction of the jobs have $\ub_j=E$ and $p_j=E$. (type E jobs)
    \item A $\gamma$ fraction of the jobs have $\ub_j=E+\epsilon$ and $p_j=E+\epsilon$ for some arbitrarily small $\epsilon>0$. (type E+ jobs)
\end{itemize}

\begin{figure}
\begin{center}
	\includegraphics[width=12cm]{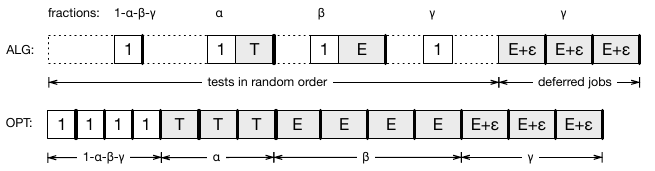}
    \caption{Worst case analysis of the algorithm \RTE.}
    \label{fig:randUB:ap}
\end{center}
\end{figure}

\subsubsection{Cost of \RTE}
Let $n$ be the total number of jobs in the instance. In the following expressions for simplification we will omit $\epsilon$.
We denote by $L:=n+T\alpha n+E\beta n$ the length of part~(1).
This means that for a job $j$ of type $0,T$ or $E$, the expected time its test starts is $(L-1-p_j)/2$ and hence its expected completion time, which is $1+p_j$ time units later, is $(L+1+p_j)/2$.
The expected objective value of \RTE can be expressed as
\begin{align}
    \textrm{ALG} =& (1-\gamma)n(n+1+T\alpha n+E\beta n)/2 \label{RTE:1}
    \\
    &+ T\alpha n/2 + E\beta n/2 \label{RTE:2}
    \\
    &+ \gamma n(n+T\alpha n+E\beta n) \label{RTE:3}
    \\
    & + E\gamma n(\gamma n+1)/2 \label{RTE:4}
\end{align}
where \eqref{RTE:1} is the sum of $(L+1)/2$ for all jobs completed in the first part, \eqref{RTE:2} is the additional part in the expected completion time for job types $T$ and $E$.  Jobs completed in the second part have all the same processing time.  The $i$-th to be completed in part (2) has completion time $L+Ei$.  Hence the total completion time of these jobs is expressed as the sum of the expressions~\eqref{RTE:3} and \eqref{RTE:4}.

\subsubsection{Cost of \OPT}
By the \emph{smallest processing time first rule}, the optimal schedule first tests and executes all type 0 jobs. Then it executes untested all type $T,E$ and $E^+$ jobs in that order. Hence the optimal objective value is stated as follows, where every other expression represents the total completion times of some job type followed by the delay these jobs induce on subsequent job types.
\begin{align*}
    \textrm{\OPT} = & (1 - \alpha - \beta - \gamma) n ((1 - \alpha - \beta - \gamma) n + 1)/2 +
     \\
  & (1 - \alpha - \beta - \gamma) n (\alpha + \beta + \gamma) n +\\
  &T \alpha n (\alpha n + 1)/2 + \\
  &T \alpha n (\beta + \gamma) n +\\
  &E \beta n (\beta n + 1)/2 + \\
  &E \beta n \gamma n +\\
  &E \gamma n (\gamma n + 1)/2.
\end{align*}

\subsubsection{Competitive ratio}
We say that fractions $\alpha,\beta,\gamma$ are \emph{valid} iff $\alpha,\beta,\gamma\geq 0$ and $\alpha+\beta+\gamma\leq 1$.
The algorithm is $T$-competitive if $T\cdot \textrm{OPT} - \textrm{ALG} \geq 0$ for all $n\geq 0$ and  all valid fractions $\alpha,\beta,\gamma$.  The costs can be written as $\ALG=\frac{n^2}2 \ALG_2 + \frac n2 \ALG_1$ and $\OPT=\frac{n^2}2 \OPT_2 + \frac n2 \OPT_1$  for
\begin{align*}
    \ALG_2 =& 1 + \gamma + \beta E + \beta \gamma E + \gamma^2 E + \alpha T + \alpha \gamma T
    \\
    \ALG_1 =& 1 - \gamma + \beta E + \gamma E + \alpha T
    \\
    \OPT_2  = &1 - \alpha^2 - 2 \alpha \beta - \beta^2 - 2 \alpha \gamma - 2 \beta \gamma - \gamma^2
    \\ &  + \beta^2 E + 2 \beta \gamma E +
     \gamma^2 E + \alpha^2 T + 2 \alpha \beta T + 2 \alpha \gamma T
 \\
    \OPT_1 =& 1 - \alpha - \beta - \gamma + \beta E + \gamma E + \alpha T.
\end{align*}
It suffices to show separately the inequalities $T\cdot \OPT_2 - \ALG_2 \geq 0$
and $T\cdot \OPT_1 - \ALG_1 \geq 0$ for all valid $\alpha,\beta,\gamma$ fractions.

We start with the first inequality, and consider the following left hand side.
\[
G = T[1 + (\beta + \gamma)^2 (E-1) + \alpha^2 (T-1) +
    2 \alpha (\beta + \gamma) (T-1) - \alpha-
 \alpha \gamma] - \gamma - 1 - E(\gamma^2 +  \beta \gamma + \beta).
\]

\subsubsection{Breaking into cases}
We want to find parameters $T,E$ with minimal $T$ such that $G(T,E,\alpha,\beta,\gamma)\geq0$ for all valid fractions, i.e.\ $\alpha,\beta,\gamma\geq 0$ with $\alpha+\beta+\gamma\leq 1$.  We call this the \emph{validity polytope} for $\alpha,\beta,\gamma$, see Figure~\ref{fig:randUBabc:ap}.  For this purpose we made numerical experiments which gave us a range where the optima could belong, namely $T\in[1.71,1.89],E\in[2.81,2.89]$.

\begin{figure}
    \centerline{\includegraphics[width=4cm]{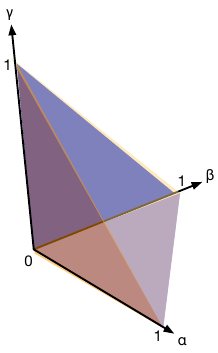}}
    \caption{Validity region for $(\alpha,\beta,\gamma)$.}
    \label{fig:randUBabc:ap}
\end{figure}
Our general approach consists in identifying values $(\alpha,\beta,\gamma)$ which are local minima for $G$.  Each of these points $(\alpha,\beta,\gamma)$ generate conditions on $T,E$ of the form $G(T,E,\alpha,\beta,\gamma)\geq 0$.  The optimal pair $(T,E)$ is then the pair with minimal $T$ satisfying all the generated conditions.

The analysis follows a partition of the validity polytope. First we consider the open region $\{(\alpha,\beta,\gamma)|0<\alpha,0<\beta,0<\gamma,\alpha+\beta+\gamma<1\}$. Then we consider the 4 open facets on the border defined by the equations $\alpha+\beta+\gamma=1, \alpha=0, \beta=0, \gamma=0$. Finally we consider the 6 closed edges that form the edges of the polytope.  Note that the vertices of the polytope $(0,0,0),(0,0,1),(0,1,0),(1,0,0)$ belong each to several edges.

\begin{itemize}
    %\item[open polytope.]
    \item Case 1: open polytope.
    The second order derivatives of $G$ in $\alpha,\beta,\gamma$ are
    \begin{align*}
        \frac{\partial^2 G}{\partial^2 \alpha} &= 2T(T-1)
        \\
        \frac{\partial^2 G}{\partial^2 \beta} &= 2T(E-1)
        \\
        \frac{\partial^2 G}{\partial^2 \gamma} &= 2T(E-1)-2E
    \end{align*}
    which are all positive in the considered $T$- and $E$-range.
    Hence a local minimum on the open polytope must be a point $(\alpha,\beta,\gamma)$ that is a root for the derivative in each of the 3 directions.
    Hence we choose $\alpha$ as the root
    \[
        \alpha = \beta - \gamma + \frac{1 + \gamma}{2 (T-1)},
    \]
    $\beta$ as the root
    \[
        \beta = \frac{1 + \gamma - 2 \gamma T}{2 T},
    \]
    and $\gamma$ as the root
    \[
        \gamma = \frac{(E (T-1) - T) (2T-1)}{E(T-1)+T}.
    \]
    For this point the condition $G\geq 0$ translates into the following condition on $T,E$.
    \begin{equation}                                        \label{eq:randUB:cond1}
        E^2 (T-1)^2 + T (2T-1)  - E T^2  \geq 0.
    \end{equation}

    %\item[facet $\alpha+\beta+\gamma=1$.]
    \item Case 2: facet $\alpha+\beta+\gamma=1$.
    In that case the derivative of $G$ in $\beta$ is $1-\alpha(E-T)$.
    This means that $G$ is linear in $\beta$, and a local minimum lies on the boundary of the triangle, which we considered open.  Hence such a local minimum will be considered in a case below.
    Note that in the degenerate case $\alpha=1/(E-T)$ the value of $G$ is independent of $\beta$, hence it is enough to consider an equivalent point on the boundary.

    %\item[facet $\gamma=0$.]
    \item Case 3: facet $\gamma=0$.
    In this case the extreme $\alpha$ value for $G$ is
    \[
        \alpha = \frac{1}{2(T-1)} - \beta,
    \]
    and then the extreme $\beta$ value for $G$ is
    \(
    \beta = 1/2T. %\frac{1}{2T}.
    \)
    For this point the condition $G\geq 0$ translates into the following condition on $T,E$.
    \begin{equation}                                        \label{eq:randUB:cond2}
        \frac{1}{T-1} + 4(T-1) - \frac{E}{T} \geq 0.
    \end{equation}

    %\item[facet $\alpha=0$.]
    \item Case 4: facet $\alpha=0$.
    In this case the extreme $\beta$ value for $G$ is
    \[
        \beta = \frac{E +\gamma E+ 2 \gamma T - 2 \gamma ET}{2T(E-1)},
    \]
    but then the second order derivative of $G$ in $\gamma$ is
    \[
        \frac{\partial^2 G}{\partial^2 \gamma} = -\frac{E^2 }{2T(E-1)}
    \]
    which is negative. Hence local minimum of this triangle is on its boundary.

    %\item[facet $\beta=0$.]
    \item Case 5: facet $\beta=0$.
    The extreme $\alpha$ value for $G$ is
    \[
        \alpha = \frac{1 + 3 \gamma - 2 \gamma T}{2(T-1)},
    \]
    and then the extreme $\gamma$ value for $G$ is
    \[
    \gamma = \frac{(2-T)(2T-1)}{4E(T-1^2) - T(5-4T(2-T))}.
    \]
    But in the considered region for $(T,E)$ the value of $\alpha+\gamma$ exceeds 1, and is therefore outside the boundaries of the triangle.

    %\item[edge $(\alpha,\beta,\gamma)=(x,1-x,0)$ for $0\leq x\leq 1$.]
    \item Case 6: edge $(\alpha,\beta,\gamma)=(x,1-x,0)$ for $0\leq x\leq 1$.
    The extreme point for $x$ is
    \[
        x = 1-\frac{1}{2T}.
    \]
    For this point the condition $G\geq 0$ translates into the following condition on $T,E$.
    \begin{equation}                                        \label{eq:randUB:cond3}
        T(T-1)-\frac34 - \frac{E}{4T} \geq 0.
    \end{equation}

    %\item[edge $(\alpha,\beta,\gamma)=(x,0,1-x)$ for $0\leq x\leq 1$.]
    \item Case 7: edge $(\alpha,\beta,\gamma)=(x,0,1-x)$ for $0\leq x\leq 1$.
    The extreme point for $x$ is
    \[
        x = \frac{2ET+2T-2T^2-2E-1}{2(E-T)(T-1)},
    \]
    which generates the following condition
    \begin{equation}                                        \label{eq:randUB:cond4}
        4 E (1 - (2-T) T^2) -(2 T (T-1) -1 )^2  \geq 0.
    \end{equation}

    %\item[edge $(\alpha,\beta,\gamma)=(0,x,1-x)$ for $0\leq x\leq 1$.]
    \item Case 8: edge $(\alpha,\beta,\gamma)=(0,x,1-x)$ for $0\leq x\leq 1$.
    Here $G$ is linear increasing in $x$, hence a local minimum is reached at $x=0$, generating the condition
    \begin{equation}                                        \label{eq:randUB:cond5}
        E(T-1) - 2 \geq 0.
    \end{equation}

    %\item[edge $(\alpha,\beta,\gamma)=(x,0,0)$ for $0\leq x\leq 1$.]
    \item Case 9: edge $(\alpha,\beta,\gamma)=(x,0,0)$ for $0\leq x\leq 1$.
    The extreme point for $x$ is
    \[
        x = \frac{1}{2(T-1)},
    \]
    generating the condition
    \begin{equation}                                        \label{eq:randUB:cond6}
         4T-5 - \frac{1}{T-1}  \geq 0.
    \end{equation}

    %\item[edge $(\alpha,\beta,\gamma)=(0,x,0)$ for $0\leq x\leq 1$.]
    \item Case 10: edge $(\alpha,\beta,\gamma)=(0,x,0)$ for $0\leq x\leq 1$.
    The extreme point for $x$ is
    \[
        x = \frac{E}{2T(E-1)},
    \]
    generating the condition
    \begin{equation}                                        \label{eq:randUB:cond7}
        4(T-1)-\frac{E^2}{T(E-1)}  \geq 0.
    \end{equation}

    %\item[edge $(\alpha,\beta,\gamma)=(0,0,x)$ for $0\leq x\leq 1$.]
    \item Case 11: edge $(\alpha,\beta,\gamma)=(0,0,x)$ for $0\leq x\leq 1$.
    The extreme point for $x$ is
    \[
        x = \frac{1}{2(ET-E-T)},
    \]
    generating the condition
    \begin{equation}                                        \label{eq:randUB:cond8}
        T-1-\frac{1}{4(ET-E-T)}   \geq 0.
    \end{equation}
\end{itemize}

In summary we want to find values $T,E$ that satisfy all conditions \eqref{eq:randUB:cond1} to \eqref{eq:randUB:cond8} and minimize $T$.  In the considered region for $T$ and $E$, the conditions \eqref{eq:randUB:cond3}, \eqref{eq:randUB:cond5}, \eqref{eq:randUB:cond6} and \eqref{eq:randUB:cond7} are satisfied.  Hence we focus on the remaining conditions, and find out that the optimal point lies on the intersection of the left hand sides of condition~\eqref{eq:randUB:cond2} and \eqref{eq:randUB:cond4}. The solutions are roots to a polynomial of degree $5$, and in only one of them $T$ is larger than the golden ratio, which it has to.  Numerically we obtain the optimal parameters $T \approx 1.7453$ and $E \approx 2.8609$.

\begin{figure}
    \centerline{\includegraphics[width=6cm]{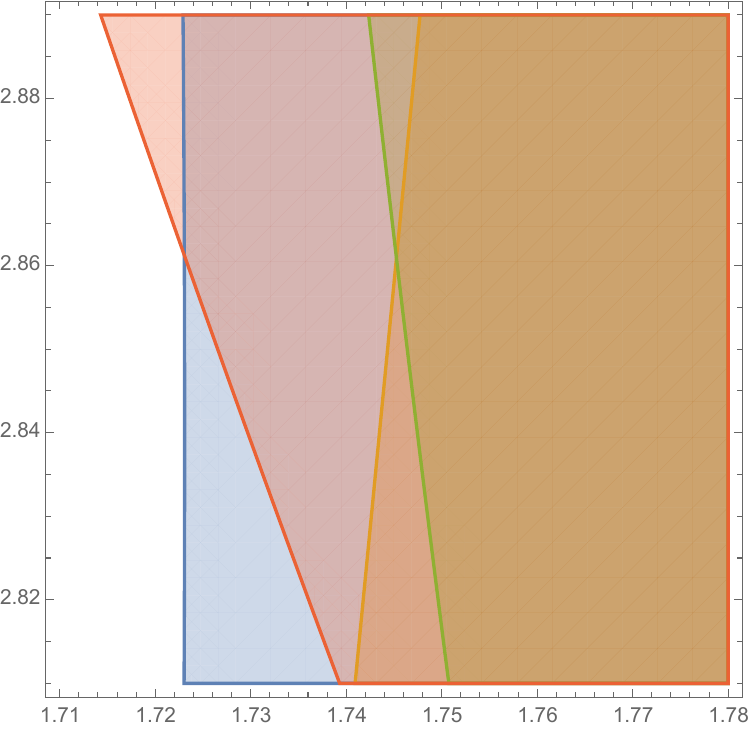}}
    \caption{Regions where conditions \eqref{eq:randUB:cond1}:blue, \eqref{eq:randUB:cond2}:orange,  \eqref{eq:randUB:cond4}:green and \eqref{eq:randUB:cond8}:red are satisfied by points $(T,E)$, with $T$ ranging horizontally and $E$ ranging vertically.}
    \label{fig:randUBcond}
\end{figure}

We conclude the proof by considering the inequality $T\cdot \OPT_1 - \ALG_1 \geq 0$ which is
\[
    \gamma (E-1) (T-1) + \beta (E-1) t + (1 - \alpha (2-T)) T  -1 - \beta E
    \geq 0.
\]
Taking the derivative of the left hand side reveals that it is decreasing in $\alpha$ and increasing in $\beta$ and $\gamma$ for the chosen values $T,E$.  Hence the expression is minimized at $\alpha=1,\beta=0,\gamma=0$, where its value is
\(
    T(T-1)-1 > 0.
\)
Therefore we have shown the following theorem.

% \THMRTE*
\begin{theorem}\label{thm:RTE}
    The competitive ratio of the algorithm \RTE is at most $1.7453$ for scheduling with testing %on a single machine
with the objective of minimizing the sum of completion times.
\end{theorem}

%!TEX root = main.tex

\subsection{Lower bound for randomized algorithms}

In this section we give a lower bound on the best possible competitive ratio of any randomized algorithm against
an oblivious adversary. We do so by specifying a probability distribution over inputs and proving a
lower bound on $E[\ALG]/E[\OPT]$ that holds for all
deterministic algorithms~$\ALG$. By Yao's principle~\cite{yao1977probabilistic,BorodinEY98} this gives the desired lower bound.

The probability distribution over inputs with $n$ jobs has a constant
parameter $0<q<1$ and
is defined as follows: Each job $j$ has upper limit~$\p_j=1/q>1$, and its processing time
$p_j$ is set to $0$ with probability $q$ and to $1/q$ with probability $1-q$.
%----------------------------------------------------------------------
\paragraph{Estimating $E[\OPT]$.}
Let $Z$ denote the number of jobs with processing time~$0$. Note that $Z$
is a random variable with binomial distribution.
The optimal schedule first tests and executes the $Z$ jobs with $p_j=0$
and then executes the $n-Z$ jobs with $p_j=1/q$ untested. Hence, the objective value
of $\OPT$ is:
$$
\frac{Z(Z+1)}{2} +Z(n-Z) + \frac{(n-Z)(n-Z+1)}{2q}.
$$
Using $E[Z]=nq$ and $E[Z^2]=(nq)^2+nq(1-q)$, we obtain
$$
E[\OPT] = \frac{n^2}{2} \left(\frac1q+3q-2-q^2\right) + O(n).
$$

\paragraph{Estimating $E[\ALG]$.}
First, observe that we only need to consider algorithms that schedule a job $j$ immediately if the job has been tested and $p_j=0$. Furthermore, we only need to consider algorithms that never create idle time before all jobs are completed.

We claim that any such algorithm satisfies $E[\ALG]\ge \frac{n^2}{2q}$ for all~$n$.
We prove this by induction on~$n$. Let $\ALG(k)$ denote the objective
value of the algorithm $\ALG$ executed for a random instance with $k$ jobs
that is generated by our probability distribution for $n=k$ (i.e.,
all $k$ jobs have $\bar{p}_j=1/q$ and $p_j$ is set to $0$ with probability $q$
and to $1/q$ otherwise).

Consider the base case $n=1$. If $\ALG$ executes job $1$ without testing, then $\ALG(1)=1/q$. If $\ALG$ tests the job and then necessarily executes it right away, since there are no other jobs, then $E[\ALG(1)]=1 + (q\cdot 0 + (1-q)\cdot (1/q)) = 1/q$. In both cases, $E[\ALG(1)]=1/q\ge \frac{n^2}{2q}$.

Now assume the claim has been shown for $n-1$, i.e., $E[\ALG(n-1)]\ge \frac{(n-1)^2}{2q}=\frac{n^2}{2q}-n/q+\frac{1}{2q}>\frac{n^2}{2q}-n/q$.
Consider the execution of $\ALG$ on an instance with $n$ jobs, and make a case distinction on how the algorithm handles the first job it tests or executes. Without loss of generality, assume that this job is job~$1$.

\begin{itemize}
\item Case 1: $\ALG$ executes job $1$ without testing (completing at time~$C_1=1/q$), or it tests jobs $1$ and then
executes it immediately independent of its processing time (with
expected completion time $E[C_1]=1+(1-q)/q=1/q$).
After the completion of job~$1$, the algorithm schedules the remaining
$n-1$ jobs, which is a random instance with $n-1$ jobs. Hence,
the objective value is $E[C_1]+ E[C_1](n-1)+E[\ALG(n-1)] =
1/q+(n-1)/q+E[\ALG(n-1)]\ge n/q + \frac{n^2}{2q}-n/q=\frac{n^2}{2q}$.

\item
Case 2: $\ALG$ tests job~$1$ and then executes it immediately if its processing time is $0$, but defers it if its processing time is~$1/q$. Assume first that if $p_1=1/q$, then $\ALG$ defers the execution of $p_1$ to the very end of the schedule. We have
$$E[\ALG(n)|p_1=0] = 1 + (n-1) + E[\ALG(n-1)] $$
 and
$$E[\ALG(n)|p_1=1/q] = n + E[\ALG(n-1)] + E[len(\ALG(n-1))] + 1/q, $$
where $len(\ALG(n-1))$ is the length of the schedule for $n-1$ jobs.
Note that every job contributes $1/q$ to the expected
schedule length no matter whether it is tested (in which case it
requires time $1$ for testing and an additional expected
$(1-q)/q$ time for processing)
or not (in which case its processing time is $1/q$ for sure).
Therefore, $E[len(\ALG(n-1))]=(n-1)/q$.
So we have:
\begin{eqnarray*}
E[\ALG(n)]&=&q(n+E[\ALG(n-1)])+(1-q)(n+n/q+E[\ALG(n-1)])\\
&=& qn+n+n/q-qn-n + E[\ALG(n-1)]\\
&=& n/q + E[\ALG(n-1)]\\
&\ge& \frac{n^2}{2q}.
\end{eqnarray*}
Finally, we need to consider the possibility that $p_1=1/q$ and $\ALG$ defers
job~$1$, but schedules it at some point during the schedule for the remaining
$n-1$ jobs instead of at the very end of the schedule. Assume that $\ALG$
schedules job~$1$ in such a way that $k$ of the remaining $n-1$ jobs
are executed after job~$1$. We compare this schedule to the schedule
where job~$1$ is executed at the very end of the schedule. Let $K$ be
the set of $k$ jobs that are executed after job~$1$ by $\ALG$.
Note that the jobs in the set $K$ can be jobs that are scheduled
without testing (and thus executed with processing time $1/q$), jobs that are tested
and executed after the execution of job $1$ (so that the expected time
for testing and executing them is $1/q$), or jobs that are tested before
the execution of job~$1$ but executed afterwards (in which case their
processing time must be $1/q$, since jobs with processing time~$0$ are
executed immediately after they are tested). Hence, moving the
execution of job $1$ from the very end of the schedule ahead of
$k$ job executions will change the expected objective value as follows:
The expected completion time of job $1$ decreases by $k/q$, and the
completion time of each of the $k$ jobs in $K$ increases by $1/q$.
Therefore, $E[\ALG(n)]$ is the same as when job~$1$ is executed at the
end of the schedule, and we get $E[\ALG(n)]\ge\frac{n^2}{2q}$ as before.
\end{itemize}

% ----------------------------------------------------------------------
\begin{theorem}\label{thm:randLB}
   No randomized algorithm can achieve a competitive ratio less than $1.6257$ for scheduling with testing %on a single machine
with the objective of minimizing the sum of completion times.
\end{theorem}
% \THMrandlb*
%
\begin{proof}
    Since we have $E[\OPT] = \frac{n^2}{2} \left(\frac1q+3q-2-q^2\right) + O(n)$
    and $E[\ALG]\ge \frac{n^2}{2q}$,
    Yao's principle~\cite{yao1977probabilistic,BorodinEY98} gives a lower bound that is arbitrarily close
    (for large enough~$n$) to
    $$
    \frac{
    1/q
    }{
    1/q+3q-2-q^2
    }
    $$
    % minimize inverse ratio:
    %1-2/u+3/u^2-1/u^3
    %' = 2-6/u+3/u^2 = 0
    %2x^2-6x+3=0
    %x = (6 +- sqrt(36-24)) /4 = 1.5 + sqrt(3)/2
    for randomized algorithms against an oblivious adversary.
    The bound is maximized for $q=1-1/\sqrt3 \approx 0.42265$,
    giving a lower bound of $1.62575$.
    \qed
\end{proof}

\section{Deterministic Algorithms for Uniform Upper Limits}\label{sec:uniform}
%!TEX root = main.tex

In this section we investigate the problem of scheduling with testing on instances in which all jobs have a uniform upper limit $\ub$. In Subsection~\ref{subsec:uniform}, we give a deterministic algorithm that achieves a ratio strictly less than~$2$. In Subsection~\ref{subsec:extreme-uniform}, we study an even further restricted class of \emph{extreme uniform} instances that consist of jobs with uniform upper limit $\ub$ and processing times in $\{0, \ub\}$. We give an algorithm with improved competitive ratio that is particularly interesting as it is near-optimal algorithm for the class of worst-case instances for deterministic algorithm from Theorem~\ref{thm:detLB} in Subsection~\ref{subsec:detLB}.

\subsection{An improved algorithm for uniform upper limits}\label{subsec:uniform}

We assume that all jobs have upper limit $\ub$. We design an algorithm with a competitive ratio strictly less than $2$ by combining \thresh, presented in Section~\ref{sec:bound2}, with a new algorithm \beat. The new algorithm \beat performs well on instances with upper limit roughly $2$, but its performance becomes worse for larger upper limits. Therefore, we employ in this case the algorithm \thresh .
%In this section we present an algorithm for instances with uniform upper limit $\ub$ that achieves a ratio strictly less than~$2$. 
%We present a new algorithm \beat that performs well on instances with upper limit roughly $2$, but its performance becomes worse for larger upper limits. Thus, in this case we employ the algorithm \thresh presented in Section~\ref{sec:bound2}.

 To simplify the analysis, we consider the limit of $\ALG(I)/\OPT(I)$ when the number~$n$ of jobs approaches infinity.
 We say that an algorithm $\ALG$ is \emph{asymptotically $\ratio$-competitive} or \emph{has asymptotic competitive ratio at most $\ratio$} if \[\lim_{n\to\infty} \sup_{I} \ALG(I)/\OPT(I) \le \ratio.\]

\begin{algo}[\beat]
The algorithm \beat balances the time testing jobs and the time executing jobs while there are untested jobs. A job is called \emph{short} if its running time is at most $\Ex = \max\{1, \ub - 1\}$, and \emph{long} otherwise. Let $\mathrm{TotalTest}$ denote the time we spend testing long jobs and let $\mathrm{TotalExec}$ be the time long jobs are executed. We iterate testing an arbitrary job and then execute the job with smallest processing time either, if it is a short job, or if $\mathrm{TotalExec}+ p_k$ is at most $\mathrm{TotalTest}$. Once all jobs have been tested, we execute the remaining jobs in order of non-decreasing processing time.
The pseudocode is shown in Pseudocode~\ref{alg:BEAT}.
\end{algo}
\begin{algorithm}
\SetAlgorithmName{Pseudocode}
	\KwIn{A set of $n$ jobs with uniform upper limit $\ub$.}
	\KwOut{A schedule of tests and executions of all jobs.}
		    TotalTest $\leftarrow$ 0\tcp*[l]{total time of executed tests of long jobs }
		    TotalExec $\leftarrow$ 0\tcp*[l]{total time of executed long jobs}
		    \While{there are untested jobs}{
		        $k \leftarrow$ tested, not executed job with minimum $p_k$\tcp*[l] {$p_k=\infty$ if no such job}
		        \uIf{$\mathrm{TotalExec} + p_k \le \mathrm{TotalTest}$}{
		        			execute $k$\; TotalExec $\leftarrow$ TotalExec + $p_k$\;
		        }
		        \Else{
		      		  $j \leftarrow$ an arbitrary untested job \; %with minimum $\bar{p}_j$ \;
		      		  test $j$\;
		     		   \uIf{$p_j \le \Ex$}{
		      			  		execute $j$\;
		      			  }
		      			\Else{
		      				    TotalTest $\leftarrow$ TotalTest + 1\;
		      			}
		      }
		}
    	execute all remaining jobs in order of non-decreasing $p_j$\;
\caption{\beat}
\label{alg:BEAT}
\end{algorithm}

We will analyze algorithm \beat in Section~\ref{subsec:beat}.
In Lemma~\ref{lem:beat_adv_algo_structure} we will make a structural observation about the algorithm schedule for a worst-case instance.
% \LEMbeatadvalgostructure*
%Consequently, the schedule produced by \beat and the optimal schedule display a clear structure, which we depict in Figure~\ref{fig:BEATschedule2}.
%
%\begin{figure}
%\includegraphics[width=11cm]{BEATalg_simple2.pdf}
%\caption{Structure of schedules produced by \beat and \OPT}
%\label{fig:BEATschedule2}
%\end{figure}
%
%
In Lemma~\ref{lem:beat_cr} we prove that the asymptotic competitive ratio of \beat for $\ub <3$ is at most
\begin{align*}
	\ratio^{BEAT} = \frac{1 + 2 (-2 + \ub) \ub + \sqrt{(1 - 2 \ub)^2 (-3 + 4 \ub)}}{2 (-1 + \ub) \ub}.
\end{align*}
This function decreases, when $\ub$ increases. Alternatively, for small upper limit we can execute each job without test. Then there is a worst-case instance where all jobs have processing time $p_j = 0$. The optimal schedule tests each job only if the upper limit $\ub$ is larger than one and executes it immediately. For $\ub<1$ this means the competitive ratio is $1$ and otherwise it is $\ub$, which monotonically increases. Thus, we choose a threshold \T$\approx 1.9338$ for $\ub$, where we start applying \beat: the fixpoint of the function $\ratio^{BEAT}$.

For upper limits $\ub > 3$, the performance behavior of \beat changes and the asymptotic competitive ratio increases. Thus, we employ the algorithm \thresh for large upper limits. Recall from Section~\ref{sec:bound2} that for $\bar p >2$ \thresh tests all jobs, executes those with $p_j \leq 2$ immediately and defers the other jobs. In Subsection~\ref{subsec:thresh}, we argue that there is a worst-case instance with short jobs that have processing time $0$ or $2$ and long jobs with processing time $\ub_j = \ub$ and that no long job is tested in an optimal solution. This allows us to prove in Theorem~\ref{thm:thresh-p} that the competitive ratio for \thresh is at most
	\[
	\ratio^{THRESH} = \left\{
                \begin{array}{ll}
                \frac{-3 + \p + \sqrt{-15 + \p (18 + \p)}}{2(\p-1)} \qquad & \textup{if }  \p\in (2,3)\\
				\sqrt{3}\approx 1.73 & \textup{if }  \p \geq 3.
			 	\end{array}
             \right.
	\]
The function for small $\p$ is a monotone function decreasing from $2$ to $\sqrt{3}$ in the limits for $\p \in (2,3)$. We choose a threshold, where we change from applying \beat to employing \thresh at \tbeatthresh $\approx 2.2948$, the crossing point of the two functions describing the competitive ratio of \beat and \thresh in $(2,3)$.
\begin{algo}\label{alg:alg-beat2}
	Execute all jobs without testing them, if the upper limit $\ub$ is less than \T$\approx 1.9338$. Otherwise, if the upper limit $\ub$ is greater than \tbeatthresh$\approx 2.2948$, execute the algorithm \thresh. For upper limits between \T and \tbeatthresh, execute the Algorithm~ \beat.
\end{algo}
The function describing the asymptotic competitive ratio depending on $\ub$ is displayed in Figure~\ref{fig:BEATcrPlot2}. Its maximum is attained at \T, which is a fixpoint. Thus we obtain the following result.

\begin{theorem}\label{thm:uniformUB}
    For scheduling with testing %on a single machine
with the objective of minimizing the sum of completion times, the asymptotic competitive ratio %$\ratio$
    of Algorithm \ref{alg:alg-beat2} on instances with uniform upper limits is $\ratio = $ \T$\approx 1.9338$,
    which is the only real root of $2 \ub^3 - 4 \ub^2 + 4 \ub - 1 - \sqrt{(1 - 2 \ub)^2 (4 \ub - 3)}$. %, which is approximately $1.9338$.
\end{theorem}
\begin{figure}
\begin{center}
	\begin{tikzpicture}[scale=1.8,every node/.style={},font=\scriptsize]
	  \draw[->] (0,0) -- (5,0) node[right,below] {$\ub$};
      \draw[->] (0,0) --  node[above,sloped] {competitive ratio}(0,2);
      \draw[dotted] (0,1.9338) node[left] {$\ratio=T_1$} -- (1.9338,1.9338) ;
      \draw[dotted]  (1,0) node[below] {$1$} -- (1,1) ;
      \draw[dotted]  (1.9338,0) node[below] {$1.93$} -- (1.9338,1.9338) ;
      \draw[dotted]  (2.3,0) node[below] {$2.29$} -- (2.3,1.905) ;
      \draw[dotted]  (3,0) node[below] {$3$} -- (3,{sqrt(3)}) ;
      \draw[domain=0:1,smooth,variable=\x] plot ({\x},{1});
      \draw[domain=1:1.93,smooth,variable=\x] plot ({\x},{\x});
      \draw[domain=1.93:2.3,smooth,variable=\x]  plot ({\x},{(1 + 2 *(-2 + \x) *\x + sqrt((1 - 2*\x)*(1 - 2*\x) *(-3 + 4*\x))))/(2 *(-1 + \x) *\x)});
     \draw[domain=2.3:3,smooth,variable=\x] plot ({\x},{(-3 +\x + sqrt(-15 + \x *(18 + \x)))/(2 *(-1 + \x))});
        \draw[domain=3:4.8,smooth,variable=\x] plot ({\x},{sqrt(3)});
			% \node[above] at (0.5,1) {$1$};
			\draw[|<->,dashed] (0,-0.4) -- node[above] {\textsc{No Tests}} (1.9338,-0.4);
			\draw[|<->|,dashed] (1.9338,-0.4) -- node[above] {\beat} (2.3,-0.4);
			\draw[<->,dashed] (2.3,-0.4) -- node[above] {\thresh} (5,-0.4);
    \end{tikzpicture}
    \caption{Competitive ratio depending on $\ub$.}\label{fig:BEATcrPlot2}
    \end{center}
\end{figure}
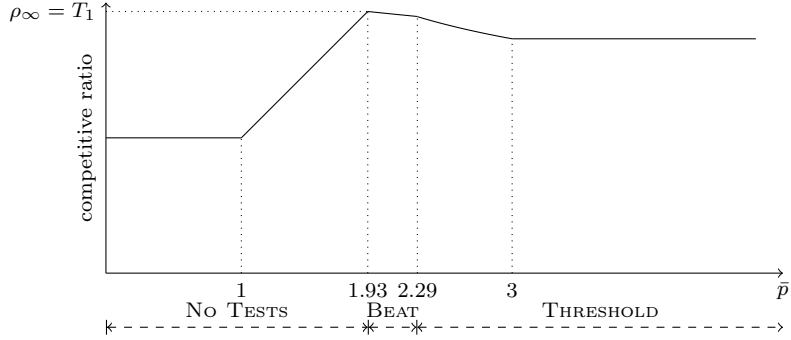

\subsubsection{Analysis of \beat}
\label{subsec:beat}%
%The algorithm \beat balances the time testing jobs and the time executing jobs while there are untested jobs. A job is called \emph{short} if its running time is at most $\Ex = \max\{1, \ub - 1\}$, and \emph{long} otherwise. Let TotalTest denote the time we spend testing long jobs and let TotalExec be the time long jobs are executed. We iterate testing an arbitrary job and then execute the job with smallest processing time either, if it is a short job, or if TotalExec $+ p_k$ is at most TotalTest. Once all jobs have been tested, we execute the remaining jobs in order of non-decreasing processing time. The pseudocode is shown in Algorithm~\ref{alg:BEAT}.
We first make a structural observation about the algorithm schedule for a worst-case instance.

\begin{lemma}\label{lem:beat_adv_algo_structure}
There are worst-case uniform instances for \beat in which
the jobs are tested in order of decreasing~$p_j$, at most one job
has $p_j\in (E, \ub)$, and all other jobs have
$p_j \in \{ 0, \Ex, \ub\}$.
\end{lemma}
\begin{proof}
Let an arbitrary worst-case instance be given. Recall that
a job is called \emph{short} if its running time is at most $\Ex = \max\{1, \ub - 1\}$,
and \emph{long} otherwise.
We first argue that the short jobs are tested last.
If not, the test and execution of some short job $j_s$ is followed
by the test of a long job~$j_l$. If $j_l$ is not executed immediately
after its test, moving the test of $j_l$ in front of the test of $j_s$ increases
the cost of the algorithm by~$1$. If $j_l$ is executed immediately
after its test, moving the test and execution of $j_l$ in front of the test
of $j_s$ increases the cost of the algorithm by $p_{j_l}-p_{j_s}>0$.
Hence, in a worst-case instance the short jobs are tested after all the tests of long jobs.

We call a long job an \emph{executed long job} if it is executed
by \beat in line 6 of Pseudocode~\ref{alg:BEAT}, and a \emph{delayed long job} or
\emph{delayed job} if it is excuted in line 15.
The long jobs have processing time larger than $\Ex\geq \ub - 1$, which means
they are not tested by \OPT. Hence, increasing the processing time of a long job
does not increase the optimal cost. For the delayed jobs, increasing their processing
time to $\ub$ increases the algorithm cost, but does not change the schedule,
so in a worst-case instance we can assume that all delayed jobs have $p_j = \ub$.

For the executed long jobs, note that no two jobs are executed without a test in
between, as their processing time is larger than one, the length of a test. We
claim that we can assume that each executed long job is tested immediately before
its execution. If not, consider an executed long job $j$ that was tested earlier
and is executed immediately after the test of another long job~$j'$.
Note that $p_{j'}\ge p_j$ and that all long jobs $j''$ executed between the
test of $j$ and the execution of $j$ satisfy $p_{j''}\le p_j$. Hence, we
can swap the tests of $j$ and $j'$ without affecting the schedule.

%\begin{figure}[!htb]
%\raisebox{0.17cm}{\makebox[3cm][r]{(a)}}%
%\quad%
%\TestL{$j$}\ExecL{1.8}{$j$}%
%\Test\Test%
%\TestL{$j'$}\ExecL{2.2}{$j'$}%
%
%\bigskip
%
%\raisebox{0.17cm}{\makebox[3cm][r]{(b)}}%
%\quad%
%\TestL{$j'$}%
%\ExecL{2.2}{$j'$}\Test\Test%
%\TestL{$j$}\ExecL{1.8}{$j$}%
%
%\caption{(a) Long job $j$ with $p_j<p_{j'}$ is executed before $j'$;
%(b) the tests (and executions) of $j$ and $j'$ have been swapped.}
%\label{fig:beatdec0}
%\end{figure}%
%
%\begin{figure}[!htb]
%%\Test\Test\ExecZ\Test\Exec{2}\Test\TestL{$j$}\Test\ExecL{2}{$j'$}
%
%%\TestL{$j$}\ExecL{1.8}{$j$}%
%%\Test\Test%
%%\TestL{$j'$}\ExecL{2.2}{$j'$}%
%%\qquad
%%\raisebox{0.17cm}{$\Rightarrow$}
%%\qquad
%%\TestL{$j'$}\ExecL{2.2}{$j'$}%
%%\Test\Test%
%%\TestL{$j$}\ExecL{1.8}{$j$}%
%
%\raisebox{0.17cm}{\makebox[3cm][r]{(a)}}%
%\quad%
%\TestL{$j$}\ExecL{1.8}{$j$}%
%\Test\Test%
%\TestL{$j'$}\ExecL{2.2}{$j'$}%
%
%\bigskip
%
%\raisebox{0.17cm}{\makebox[3cm][r]{(b)}}%
%\quad%
%\TestL{$j'$}%
%\Test\ExecL{2.2}{$j'$}\Test%
%\TestL{$j$}\ExecL{1.8}{$j$}%
%
%%% \bigskip
%%% 
%%% \raisebox{0.17cm}{\makebox[3cm][r]{(c)}}%
%%% \quad%
%%% \Test%
%%% \TestL{$j'$}\ExecL{2.2}{$j'$}\Test%
%%% \TestL{$j$}\ExecL{1.8}{$j$}%
%\caption{(a) Long job $j$ with $p_j<p_{j'}$ is executed before $j'$;
%(b) the tests and executions of $j$ and $j'$ have been swapped,
%and the execution of $j'$ has moved after the test of a long
%delayed job.}
%\label{fig:beatdec}
%\end{figure}%

\begin{figure}[!htb]
\begin{minipage}[t]{0.48\textwidth}
\centering 
\raisebox{0.20cm}{\makebox[2em][r]{(a)}}
\quad%
\TestL{$j$}\ExecL{1.8}{$j$}%
\Test\Test%
\TestL{$j'$}\ExecL{2.2}{$j'$}%

\bigskip

\raisebox{0.20cm}{\makebox[2em][r]{(b)}}
\quad%
\TestL{$j'$}%
\ExecL{2.2}{$j'$}\Test\Test%
\TestL{$j$}\ExecL{1.8}{$j$}%

\caption{(a) Long job $j$ with $p_j<p_{j'}$ is executed before $j'$;
(b) the tests (and executions) of $j$ and $j'$ have been swapped.}
\label{fig:beatdec0}
\end{minipage}
\hfill
\begin{minipage}[t]{0.48\textwidth}
\centering 
\raisebox{0.20cm}{\makebox[2em][r]{(a)}}%
\quad%
\TestL{$j$}\ExecL{1.8}{$j$}%
\Test\Test%
\TestL{$j'$}\ExecL{2.2}{$j'$}%

\bigskip

\raisebox{0.20cm}{\makebox[2em][r]{(b)}}%
\quad%
\TestL{$j'$}%
\Test\ExecL{2.2}{$j'$}\Test%
\TestL{$j$}\ExecL{1.8}{$j$}%

\caption{(a) Long job $j$ with $p_j<p_{j'}$ is executed before $j'$;
(b) the tests and executions of $j$ and $j'$ have been swapped,
and the execution of $j'$ has moved after the test of a long
delayed job.}
\label{fig:beatdec}
\end{minipage}
\end{figure}%

Next, we claim that we can assume that the executed long jobs are
tested in order of decreasing processing times. If not, there must be
an executed long job~$j$ that precedes an executed long job~$j'$
(potentially with some tests of delayed long jobs in between) such that
$p_j<p_{j'}$. Swap the tests of $j$ and~$j'$. If job $j'$ is still
executed immediately after its test in the new position (see
Figure~\ref{fig:beatdec0}), the cost of the algorithm increases
by $p_{j'}-p_j$. If job $j'$ is executed only after a further test
of a delayed job (see Figure~\ref{fig:beatdec}; this happens if $\mathrm{TotalExec}+p_{j'}>\mathrm{TotalTest}$
holds after testing~$j'$),
%we can swap the test of $j'$ with
%the test of the delayed job that is tested just before the execution
%of~$j'$ (see Figure~\ref{fig:beatdec}), which does not affect the
%schedule as argued in the previous paragraph.
%In this case,
the cost of the algorithm increases by
$1+(p_{j'}-p_j)$. Note that the execution of $j'$ cannot move behind
two or more tests of delayed jobs because $\Ex< p_j<p_{j'}\le\ub$ implies
$p_{j'}-p_j<1$. As the cost of the algorithm increases in both cases while
the optimal cost remains unchanged,
the executed long jobs must indeed be tested in order of decreasing processing
times in a worst-case instance.

\begin{figure}[!htb]
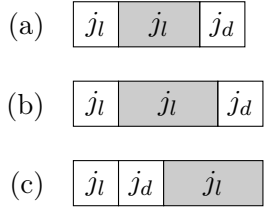

\raisebox{0.17cm}{\makebox[3cm][r]{(a)}}%
\quad%
\TestL{$j_l$}\ExecL{1.8}{$j_l$}\TestL{$j_d$}

\bigskip

\raisebox{0.17cm}{\makebox[3cm][r]{(b)}}%
\quad%
\TestL{$j_l$}\ExecL{2.2}{$j_l$}\TestL{$j_d$}

\bigskip

\raisebox{0.17cm}{\makebox[3cm][r]{(c)}}%
\quad%
\TestL{$j_l$}\TestL{$j_d$}\ExecL{2.2}{$j_l$}

\caption{(a) The execution of the last long executed job $j_l$ is followed
by the test of a delayed job~$j_d$;
(b) increasing $p_{j_l}$ increases the cost of the algorithm;
(c) if the execution of $j_l$ moves after the test of $j_d$,
the increase in cost is even larger.}
\label{fig:lem13case1}
\end{figure}

\begin{figure}[!htb]
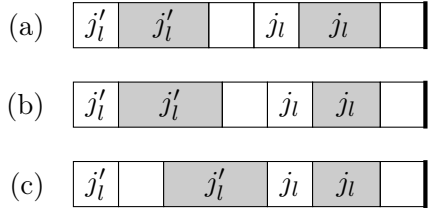

\raisebox{0.17cm}{\makebox[3cm][r]{(a)}}%
\quad%
\TestL{$j_l'$}\ExecL{2.0}{$j_l'$}\Test%
\TestL{$j_l$}\ExecL{1.8}{$j_l$}\Test

\bigskip

\raisebox{0.17cm}{\makebox[3cm][r]{(b)}}%
\quad%
\TestL{$j_l'$}\ExecL{2.3}{$j_l'$}\Test%
\TestL{$j_l$}\ExecL{1.5}{$j_l$}\Test

\bigskip

\raisebox{0.17cm}{\makebox[3cm][r]{(c)}}%
\quad%
\TestL{$j_l'$}\Test\ExecL{2.3}{$j_l'$}%
\TestL{$j_l$}\ExecL{1.5}{$j_l$}\Test

\caption{(a) The execution of the last long executed job $j_l$ is followed
by the test of a short job;
(b) increasing $p_{j_l'}$ and decreasing $p_{j_l}$ increases the cost of
the algorithm;
(c) if the execution of $j_l'$ moves after the test of a delayed
long job, the increase in cost is even larger.}
\label{fig:lem13case2}
\end{figure}
Now, we want to show that we can also assume that
the processing times of all the executed long jobs (with at most one exception)
are equal to~$\ub$.
Consider the last executed long job~$j_l$, and assume that
$p_{j_l}<\ub$ (otherwise, all executed long jobs have
processing time~$\ub$).
Case 1: If $j_l$ is followed by the test of a delayed
long job~$j_d$, we increase $p_{j_l}$ to $\ub$, an increase
of less than~$1$. After this increase, $j_l$ will either
still be executed immediately after its test, or it will
be executed after the test of~$j_d$ (see Figure~\ref{fig:lem13case1}).
The cost of the algorithm has thus increased by at least~$\ub-p_{j_l}>0$.
Case 2:
If the execution of the last executed long job is followed by the
test of a short job and there is at least one other executed long job
with processing time strictly less than~$\ub$, we proceed as follows:
We shift processing time from the last executed long job $j_l$ to
the one before, say $j_l'$, until either $p_{j_l'}=\ub$ or
$p_{j_l}=\Ex$. This increases the completion time of the first of the two jobs
(the execution of that job may potentially also move after the test of a delayed
long job), but does not change the completion time of any other job
(see Figure~\ref{fig:lem13case2}). If $p_{j_l}$
becomes equal to $\Ex$, the job $j_l$ becomes a short job, but the
schedule of the algorithm does not change.
Thus, in both cases the cost of the algorithm can be increased
while keeping the optimal cost unchanged, a contradiction to
the instance being a worst-case instance.
Hence, neither Case~1 nor Case~2 can apply in a worst-case instance,
and therefore we have at most one
executed long job with processing time strictly less than~$\ub$, and that
job (if it exists) is tested last among all long jobs.

Finally we observe that both the algorithm and the optimal schedule test all
short jobs with $p_j\in [0, \ub - 1]$ independent of their actual processing
time. Also the execution order of the algorithm and the optimal schedule solely
depend on the ordering of the processing times. Therefore
Lemma~\ref{lem:linear-pj} implies that we can assume that the short jobs have processing
times either $0$ or $\ub - 1$. Next, observe that increasing the processing
times of all short jobs with processing times in $[ \ub -1, \Ex]$ to $\Ex$
does not change the optimal cost as \OPT can execute these jobs untested
(recall that a job with $p_j=\ub-1$ takes time $\ub$
no matter whether it is tested and executed, or executed untested).
It increases the algorithm cost, however. Thus, we can
assume that in a worst-case instance all short jobs have $p_j \in \{0, \Ex\}$.
It is also clear that in a worst-case instance the short jobs are tested in
order of decreasing processing times by the algorithm, and hence all jobs are tested
in order of decreasing processing times (first the long jobs with processing time
$\ub$, then possibly the one long job with processing time between $E$ and $\ub$,
and finally the short jobs).
\qed
\end{proof}
Consequently, the schedule produced by \beat consists of the following
parts (in this order), see also Figure~\ref{fig:BEATschedule3}:
\begin{itemize}
\item The tests of the $\lambda$ fraction of jobs, that are long jobs, interleaved with executions of the $\eta$ fraction of all jobs, that are also long jobs and that are executed during the ``while there are untested jobs'' loop.
\item The tests and immediate executions of the short jobs, which is a $\sigma = 1 - \lambda$ fraction of all jobs. Let $\delta$ be the fraction of short jobs with $p_j = \Ex$.
\item The executions of the $\psi = \lambda - \eta$ fraction of jobs, that are delayed long jobs, in the ``execute all remaining jobs'' statement.
\end{itemize}
\medskip
\OPT consists of the following parts (in this order), see also Figure~\ref{fig:BEATschedule3}:
\begin{itemize}
\item The tests and immediate executions of the $(1- \delta)\sigma$ fraction of jobs that are short and have processing time~$0$.
\item The untested executions of the $\delta \sigma$ fraction of jobs which are short and have $p_j=E$ and the $\lambda$ fraction of jobs that are long.
\end{itemize}
\begin{figure}
\centering 
% \begin{tabular}{r|c|c|c|c|}
% \cline{2-4}
%     BEAT: &  $\lambda n$ long jobs tested &
%     $ \sigma n$ short jobs  &
%     $\psi n$ delayed long jobs
%     \\
%     & $\eta n$ long jobs executed  &
%     tested and executed &
%     jobs executed \\
%     \cline{2-4}
%     \\
%     \cline{2-3}
%     OPT: &
%     $(1-\delta)n$ short jobs &
%     $\delta \sigma n$ short jobs and $\lambda n$ long jobs
%     \\
%     &
%     tested and executed &
%     executed untested
%     \\
%     \cline{2-3}
% \end{tabular}
\includegraphics[width=11cm]{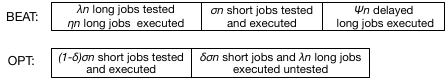}
\caption{Structure of schedules produced by \beat and \OPT.}
%xtof: habe " (with fractions of jobs indicated)" entfernt, weil n als Faktor vorkommt.
\label{fig:BEATschedule3}
\end{figure}

We note that TotalTest has value $\lambda n$ when all long jobs are tested, so the total execution time in Phase~1, which is at least $\ub (n \eta -1) + E$ by Lemma~\ref{lem:beat_adv_algo_structure}, cannot exceed $\lambda n$. As long jobs have $p_j > E\geq 1$, there are always at least as many long jobs tested as are executed. Thus, TotalExec never decreases below $\mathrm{TotalTest} - \ub$, as then some job can be executed. Hence, we have
\begin{equation}\label{eq:nalpha}%
	\ub \eta \le \lambda + O(1/n) <  \bar p \eta + O(1/n).
\end{equation}
Furthermore, we have $\lambda = \eta + \psi$, which yields
\begin{equation}\label{eq:nbeta}
	\psi \le \left(1 - 1 / \ub \right) \lambda + O(1/n).
\end{equation}
%because TotalTest has value $\lambda$ when all long jobs are tested, so the total execution time in Phase~1 cannot exceed $\lambda$. For $n \rightarrow \infty$ we get equality. Furthermore, we have $\lambda = \eta + \psi$, which yields for $n \rightarrow \infty$
%\begin{equation}\label{eq:nbeta}
%	\psi = \left(1 - 1 / \ub\right) \lambda.
%\end{equation}
%
	We first consider the algorithm schedule.
\begin{lemma}\label{lem:beat_algo_cost}
	For %uniform upper limit $\ub \in [$\T, \tbeatthresh$]$ and
	a fraction $\delta \in [0,1]$ of short jobs with processing time $p_j = \Ex$, we can bound the algorithm cost by
\begin{align*}
	 \ALG	\le &		\frac{n^2}{2} \left[ \lambda^2 \left(\ub + 2  - \frac{1}{\ub} \right)
	 						+ \sigma^2 ((1 + \Ex) (2 \delta - \delta^2) + (1 - \delta)^2) \right.	\\
	 				 &		\qquad \left.	+ 2\lambda \sigma \left(2 + \left(1 - \frac{1}{\ub}\right)(1 + \Ex \delta)\right)\right] + O(n).
\end{align*}
\end{lemma}
\begin{proof}
	There is an $\eta$ fraction of jobs completed in the first part, each executed when TotalExec$+ p_j \le$ TotalTest in the algorithm. Thus, the completion time of the $i$-th such job is at most $2i \ub +1$. The sum of these completion times is $\ub \eta^2 n^2 + O(n)$.
A fraction of $\delta \sigma$ jobs is short and has
$p_j = E$.
They are executed before the other $(1 - \delta)\sigma$ fraction of jobs with $p_j = 0$ is executed. This means the completion times of the short jobs contribute
	\begin{align*}
		 \frac{n^2}{2}\left[(1 + \Ex) \delta^2 \sigma^2 + (1-\delta)^2 \sigma^2 + 2 (1 + \Ex) \delta \sigma (1-\delta) \sigma \right]+ O(n)	\\
		 =	\frac{n^2}{2} \left[ \sigma^2 ((1 + \Ex) (2 \delta - \delta^2) + (1 - \delta)^2) \right] + O(n) .
	\end{align*}
Additionally there is an $\psi$ fraction of jobs, which are executed at the end of the schedule, each with processing time $\ub$. Thus their contribution to the algorithm cost is $\ub \psi^2 n^2/2 + O(n)$. The execution of the fraction $\sigma$ of short jobs starts latest at time $n \lambda + \ub n \eta$, and the execution of the fraction $\psi$ of jobs is delayed by at most $n \lambda + \ub n \eta + (1 + \Ex \delta)n \sigma$. Thus, the total objective value of \beat is at most:
\begin{align*}
\ALG&	 \le		\frac{n^2}{2} \left[ 2 \ub \eta^2
						 + \sigma^2 ((1 + \Ex) (2 \delta - \delta^2) + (1 - \delta)^2)
						 + \ub \psi^2 \right. 	\\
					&	\qquad  \left. + 2(\lambda + \ub \eta)\sigma + 2(\lambda + \ub \eta + (1 + \Ex \delta)\sigma)\psi \right] + O(n).
\end{align*}
By \eqref{eq:nalpha} and \eqref{eq:nbeta}, we know that $\eta \leq \lambda/\ub + O(1/n)$ and $\psi \le \left(1-1/\ub\right) \lambda + O(1/n)$. Together with $\eta + \psi = \lambda$, this yields the desired bound.
\qed
\end{proof}

\begin{lemma}\label{lem:beat_cr}
	For uniform upper limit $\ub \in [$1.5, 3$]$, the asymptotic competitive ratio of \beat is at most
	\[
		\frac{1 + 2 (-2 + \ub) \ub + \sqrt{(1 - 2 \ub)^2 (-3 + 4 \ub)}}{2 (-1 + \ub) \ub}.
	\]
\end{lemma}
\begin{proof}
	We bounded the algorithm cost in Lemma~\ref{lem:beat_algo_cost} and thus first consider the optimal cost. In \OPT, first a fraction $(1 - \delta) \sigma$ of the short jobs is tested and executed with processing time $0$. Then the remaining fraction $\delta \sigma$ of short jobs is executed with processing time $\ub$ without test. Thus their contribution to the sum of completion times is
\[	 \frac{n^2}{2} \left[\sigma^2\left((1-\delta)^2 + 	\ub \delta^2 	+ 	2 \delta (1-\delta)\right)\right] + O(n)
\	=	\frac{n}{2}\left[\sigma^2 \left((\bar p - 1) \delta^2 + 1\right)\right] + O(n).\]
	 All long jobs are executed untested at the end of the schedule and take $\ub$ time units. Their sum of completion times is $\ub \lambda^2 n^2/ 2 + O(n)$ and they are each delayed by $\sigma n(1 + (\ub -1) \delta)$), giving:
$$
\OPT =	\frac{n^2}{2} \left[ \lambda^2 \ub
				+ 	\sigma^2 ((\ub - 1) \delta^2 + 1)
				+ 2 \lambda \sigma (1+ (\ub -1) \delta)\right] + O(n).
$$
Then the asymptotic competitive ratio $\ratio$ for upper limit $\bar p$ in $[1.5, 3]$
$$
\ratio = \frac{\lambda^2 \left(\ub + 2 - \frac{1}{\ub}\right)
	+ 	\sigma^2 ((1 + \Ex) (2 \delta - \delta^2) + (1 - \delta)^2)
	+ 2\lambda \sigma (2 + \left(1 - \frac{1}{\ub}\right)(1 + \Ex \delta ))}
{\ub \lambda^2
	+ \sigma^2 ((\ub -1) \delta^2 + 1)
	+ 2\lambda \sigma (1+ (\ub -1) \delta)}.
$$
For $\sigma = 0$ or $\lambda = 0$ this fulfills the claim. For the other values we set $\sigma = \alpha \lambda$ so the ratio becomes:
$$
\frac{\ub + 2 - \frac{1}{\ub}
	+ \alpha^2 ((1 + \Ex) (2 \delta - \delta^2) + (1 - \delta)^2)
	+ 2 \alpha (2 + \left(1 - \frac{1}{\ub}\right)(1 +  \Ex \delta))}
{\ub +
\alpha^2 ((\ub -1) \delta^2 + 1)
+ 2\alpha (1+ (\ub -1)\delta)}.
$$
We take the term to Mathematica to find the best bounds for it. For the case $1.5 < \ub < 2$ we show that the adversary chooses $\delta = 0$ and $\alpha$ such that the first derivative in $\alpha$ equals $0$. Otherwise, in the case $2 \leq \ub \leq 3$, we show for $\delta = 0$ that we get exactly the same expression as for $\ub < 2$. We prove the adversary chooses this case, which means the competitive ratio is bounded by the following function
\begin{align*}
	&&\frac{1 + 2 (-2 + \ub) \ub + \sqrt{(1 - 2 \ub)^2 (-3 + 4 \ub)}}{2 (-1 + \ub) \ub}.
	\end{align*}
    \qed
\end{proof}

\subsubsection{Analysis of \BOUND for uniform $\bar{p}$}
\label{subsec:thresh}

In this section we analyze Algorithm \BOUND (see Section~\ref{sec:bound2})
for instances with uniform upper limit $\bar{p}>2$ and derive a competitive ratio as a function of $\p$.

Recall that for $\bar{p}>2$, \BOUND tests all jobs. It executes a job immediately
if $p_j\le 2$, and defers it otherwise. We have proved in Lemma~\ref{lem:wc-short-jobs} that we may assume that all jobs with $p_j\leq 2$ have execution times either $0$ or $2$. We also argued that in a worst case, \BOUND tests first all long jobs, i.e., jobs $j$ with $p_j>2$, then follow the short jobs with tests (first length-$2$ jobs and then length-$0$ jobs), and finally \BOUND executes the deferred long jobs in increasing order of processing times.

An optimum solution tests a job $j$ only if $p_j + 1 < \p$. We show next that such long jobs to be tested in an optimal solution do not exist.

%\LEMnoMediumJobs*
\begin{lemma}\label{lem:no-medium-jobs}
	There is a worst-case uniform instance with short jobs that have processing times $0$ or $2$ and long jobs with processing time $p_j=\p$. Furthermore, none of the long jobs is tested in an optimal solution.
\end{lemma}
\begin{proof}
	Consider an instance with short jobs that have processing times $0$ or $2$ (Lemma~\ref{lem:wc-short-jobs}). We may increase the processing time of untested long jobs to their upper limit $\p$ without changing the optimal schedule. This cannot decrease the worst-case ratio as the algorithm's objective value can only increase.

	It remains to consider the long jobs that are tested by an optimal solution. We show that we may assume that those do not exist. This is trivially true if $2 < \p < 3$. Then testing a long job $j$ costs $1+p_j>3$ which is greater than running the job untested at $\p<3$, and thus, an optimal solution would never test it.

	Assume now that $\p\geq 3$. \BOUND schedules any long job {\em after} all short jobs; first it runs long tested jobs with total execution time $1+p_j<\p$ in non-decreasing order of $p_j$ and then the untested jobs with execution time $\p$. As all untested jobs have processing time $p_j=\p$, we may assume that the algorithm and the optimum schedule long jobs in the same order. Reducing the processing times of all tested long jobs to $2+\varepsilon$ for infinitesimally small $\varepsilon>0$ does not change the schedule for any of the two algorithms, and thus, by Proposition~\ref{prop:change-indep-deltap}, the ratio of the objective values of the algorithm and the optimum does not decrease.

%	First, we argue that we can reduce the processing times of all such jobs to $2+\varepsilon$ for infinitesimally small $\varepsilon>0$ without decreasing the ratio of the objective values of the algorithm and the optimum. Let $j$ be the first tested long job with $p_j>2+\varepsilon$. Reducing the processing time of $j$ by $\delta=p_j-(2+\varepsilon)$ reduces the objective values of the algorithm and of the optimum by the same amount  and thus can only increase the ratio. We can repeat this argument, and may assume that all tested long jobs have length $p_j=2+\varepsilon$.

	Now, we argue that reducing the processing times of tested long jobs from $2+\varepsilon$ to~$2$ (thus making them short jobs) does not affect the optimal objective value, because~$\varepsilon$ is infinitesimally small, and can only increase the objective value of the algorithm. Consider the first long job that is tested by the optimum and the algorithm, say job~$\ell$. Consider the worst-case schedule of our algorithm for the new instance in which~$\ell$ is turned into a short job with effectively the same processing time. The job $\ell$ is tested and scheduled just before the short jobs with $p_j=0$ instead of after them.
		%When we turn~$\ell$ into a short job with effectively the same processing time, the algorithm schedules that job just before the short jobs with $p_j=0$ instead of after them.
		Let $a$ be the number of those short jobs. Then this change in $p_{\ell}$ to $2$ improves the completion time of job $\ell$ by $a$ and increases the completion time of $a$ jobs by $2$, so the net change in the objective value of the algorithm is $2a-a=a\ge 0$. The argument can be repeated until no tested long jobs are~left.
		\qed
	\end{proof}

%\THMthresh*
\begin{theorem}\label{thm:thresh-p}
	For scheduling with testing %on a single machine
jobs with uniform upper limit $\p>2$ with the objective of minimizing the sum of completion times, Algorithm \BOUND has an asymptotic competitive ratio at most
	\[
	\ratio = \left\{
               \begin{array}{ll}
                \frac{-3 + \p + \sqrt{-15 + \p (18 + \p)}}{2(\p-1)} \qquad & \textup{if }  \p\in (2,3)\\
				\sqrt{3}\approx 1.73 & \textup{if }  \p \geq 3.
			 	\end{array}
             \right.
	\]
	The function for small $\p$ is a monotone function decreasing from $2$ to $\sqrt{3}$ in the limits for $\p \in (2,3)$.
\end{theorem}
\begin{proof}
Consider a worst-case instance according to Lemma~\ref{lem:no-medium-jobs}. Let $\alpha n$ denote the number of short jobs of length $0$, let $\beta n$ be the number of short jobs of length $2$, and let $\gamma n$ be the number of long jobs with $p_j=\p$. There are no other jobs, so $\alpha+\beta+\gamma=1$. Recall, that we may assume that \BOUND's schedule is as follows: first $\gamma n$ tests, $\beta n$ tests and executions of length-$2$ jobs, then tests and executions of $\alpha n$ length-$0$ jobs, followed by the execution of long jobs with $p_j=\p$. 	The objective value is
	\begin{align}
		% \ALG= c(a+b+c) + b^2/2 \cdot 3 + 3b(a+c) + a^2/2 + ac + c^2/2 \p + O(n)
		\ALG= n^2\left(\gamma (\alpha + \beta + \gamma) + \frac{\beta^2}{2}\cdot 3 + 3\beta (\alpha + \gamma) + \frac{\alpha^2}{2} + \alpha \gamma + \frac{\gamma^2}{2} \cdot \p \right)+ O(n).\label{eq:alg}
	\end{align}

	To estimate the objective value of an optimal solution, we distinguish two cases for the upper limit $\p$.

	\paragraph{Case: $\p > 3$.} In this case, an optimal solution would test all short jobs, first the length-$0$ jobs and then the length-$2$ jobs. Then follow all long jobs without testing them (Lemma~\ref{lem:no-medium-jobs}). 	Using the above notation, we have an optimal objective value
	\begin{align*}
		%OPT &= a^2/2 + a(b+c) + b^2/2\cdot 3 + 3bc + {c^2}/{2} \cdot \p + O(n).
		\OPT &= n^2\left(\frac{\alpha^2}{2} + \alpha(\beta+\gamma) + \frac{\beta^2}{2}\cdot 3 + 3\beta\gamma + \frac{\gamma^2}{2} \cdot \p\right) + O(n).
	\end{align*}
%	We observe that the  objective value of the optimum and the algorithm have the same dependency on the upper limit $\p$, and thus, by Proposition~\ref{prop:change-indep-deltap} the ratio of both is maximized for the smallest choice of $\p>3$ which is $\p=3+\eps$ for an infinitesimally small $\eps>0$. Using this choice of $\p$ and $\gamma = 1-\alpha-\beta$ the objective values simplify for $\eps\rightarrow 0$ to
%	\begin{align*}
%		\ALG&=  \frac{n^2}{2} (5 - 6 \alpha + 2 \alpha^2) + \frac{n^2}{2} 2(-1 + 2 \alpha) \beta + O(n) \qquad \textup{ and}\\
%	%\end{align*}
%	%and
%	%\begin{align*}
%			OPT &= \frac{n^2}{2} (3 - 4 \alpha + 2 \alpha^2) + O(n).
%	\end{align*}
%	%
%	The asymptotic competitive ratio is for any $\p >3$ bounded by
%	\[
%	 \ratio = \frac{5 - 2 \beta + 2 \alpha (-3 + \alpha + 2 \beta)}{3 - 4 \alpha + 2 \alpha^2},
% 	\]
%	which has its maximum at $\sqrt{3}$ for $\alpha = (3 - \sqrt{3})/2$ and $\beta=(\sqrt{3}-1)/2$.

	Using $\gamma = 1-\alpha-\beta$, the asymptotic competitive ratio for any $\p >3$ can be bounded by
	\[
	\frac{2 - \alpha^2 - 2 \alpha \beta + (4 - 3 \beta) \beta + \p (-1 + \alpha + \beta)^2}{-\alpha^2 + \alpha (2 - 6 \beta) - 3 (-2 + \beta) \beta + \p (-1 + \alpha + \beta)^2},
	\]
which has its maximum at $\sqrt{3}$ for $\alpha = (3 - \sqrt{3})/2$ and $\beta=(\sqrt{3}-1)/2$.

\paragraph{Case: $\p \leq 3$.} In this case, an optimal solution tests only short jobs with $p_j=0$ and executes all other jobs untested, also short jobs with $p_j=2$. The value of an optimum schedule is
		\begin{align*}
			\OPT &= n^2\left( \alpha^2/2 + \alpha(\beta + \gamma) + \p \cdot \beta^2/2  + \p \cdot \beta \gamma + \p \cdot \gamma^2/2 \right) + O(n).
		\end{align*}
	With the value of \BOUND's solution given by Equation~\eqref{eq:alg}, the asymptotic competitive ratio is
	\[
	 \ratio = \frac{\alpha^2 + 3 \beta^2 + 8 \beta \gamma + \alpha (6 \beta + 4 \gamma) + \gamma^2 (2 + \p)}{\alpha^2 +  2 \alpha (\beta + \gamma) + (\beta + \gamma)^2 \p}\,.
 	\]

	Using Mathematica we verify that this ratio has its maximum at the desired value
	\[
	\frac{-3 + \p + \sqrt{-15 + \p (18 + \p)}}{2(\p-1)}\,.
	\]
	\qed
\end{proof}

%!TEX root = main.tex

\subsection{A nearly optimal deterministic algorithm for extreme uniform instances}\label{subsec:extreme-uniform}

We present a deterministic algorithm for the restricted class of \emph{extreme uniform} instances, that is almost tight for the instance that yields the deterministic lower bound. An \emph{extreme uniform} instance consists of jobs with uniform upper limit $\ub$ and processing times in $\{0, \ub\}$. Our algorithm \ute\ requires a parameter $\rho\geq 1$ and attains
%asymptotic competitive ratio~$\ratio \approx 1.8668$
competitive ratio~$\mratio \approx 1.8668$
for this class of instances when setting the algorithm parameter $\rho$ accordingly.

\begin{algo}[\ute]
Let parameter $\rho\geq 1$ be given. If the upper limit $\ub$ is at most~$\rho$, then all jobs are executed without test. Otherwise, all jobs are tested. The first $\max\{0,\beta\}$ fraction of the jobs are executed immediately after their test. The remaining fraction of the jobs are executed immediately after their test if they have processing time $0$ and are delayed otherwise, see Figure~\ref{fig:UTE2}.  The parameter $\beta$ is defined as
\begin{equation} \label{eq:ute:b2}
    \beta = \frac{1 - \ub + \ub^2 - \rho + 2 \ub \rho - \ub^2 \rho}
                {1 - \ub + \ub^2 - \rho + \ub \rho}.
\end{equation}
\end{algo}
The choice of $\beta$ will become clear in the analysis of the algorithm.
% \begin{algorithm}
%   \KwIn{A set of $n$ jobs with uniform upper limit $\ub$ and processing time~$p_j \in \{0, \ub\}$.}
%   \KwOut{A schedule of tests and executions of all jobs.}
%   \eIf{$\ub \leq \rho$}{
%        execute all jobs with without test\;
%   }
%   {
%       set $\beta=\max\{f(\ub),0\}$\;
%       test and execute the first~$\beta$ fraction of jobs immediately after their test\;
%       \For{the remaining $1 - \beta$ fraction of jobs $j$}{
%               test job~$j$\;
%               \eIf{$p_j =0$}{
%                       execute job $j$\;
%               }{
%                       delay job $j$\;
%               }
%       }
%       execute all delayed jobs\;
%   }
% \caption{\ute}
% \label{alg:ute}
% \end{algorithm}

\begin{figure}
\begin{center}
    \includegraphics[width=12cm]{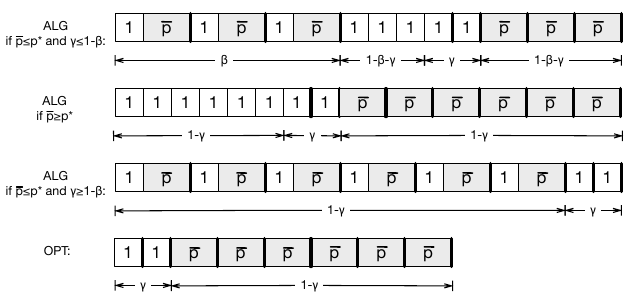}
    \caption{The schedule produced by UTE and the optimal schedule.}
    \label{fig:UTE2}
\end{center}
\end{figure}

% We have optimized the parameter $b$ and choose a complicated function depending on the upper limit $\ub$ with the following properties: The function has a fixpoint at $\ratio$ and then monotonically decreases until turning $0$ roughly for $\ub = 3.1$.
\begin{theorem}\label{thm:ext-uniform-UB}
   For scheduling with testing to minimize the sum of completion times, the competitive ratio of \ute on extreme uniform instances is at most~$\rho = \frac{1 + \sqrt{3 + 2 \sqrt{5}}}{2} \approx 1.8668$.
\end{theorem}
% \THMuteratio*
%
%
\begin{proof}
    If the upper limit $\ub$ is at most $\rho$, by Lemma~\ref{lem:all_jobs_greater_c} the algorithm has competitive ratio $\ub$, which fulfills the claim. Thus, we assume in the following $\ub \geq \rho$.
    An instance is defined by the job number $n$, an upper limit $\ub$ and a fraction $\gamma$ such that
the first $1-\gamma$ fraction of the jobs tested by \ute have processing time $\ub$, while the jobs in the remaining $\gamma$ fraction have processing time $0$.  The algorithm chooses $\beta$ so as to have the smallest ratio $\rho$.

    With the chosen fixed value of $\rho$, the value $\beta$ from equation~\eqref{eq:ute:b2} is a decreasing function in $\ub$ for $\ub\geq\rho$. Hence there is a threshold value $p^*$ such that $\beta(\ub)\leq 0$ for all $\ub \geq p^*$, which is
    \[
            p^* :=
            \frac{2 \rho + \sqrt{4 \rho-3} - 1}
                    {2 (\rho-1)} \approx 2.7961.
    \]

    As in previous proofs, we start to analyze the ratio only for the $n^2$ dependent part of the costs of \ute and OPT.    We distinguish three cases, depending on the ranges of $\ub$ and $\gamma$.

% ----------------------------------------------------------------------------------------------------------
    \paragraph{Case $\rho \leq \ub \leq p^*$ and $\gamma \leq 1-\beta$.}
    Consider for now $\beta$ and $\rho$ as some undetermined parameters which will be optimized in the analysis of this case.
        The optimal cost is
\[
\OPT = \gamma^2/2 + \ub(\gamma-1)^2/2 + \gamma(1 - \gamma)
\]
    while the cost of \ute is
    \begin{align*}
        \ALG=& (\ub + 1)\beta^2/2+ \\
               & \ub (1- \beta - \gamma)^2/2 + \\
               & \gamma^2/2+ \\
               & (1 - \gamma + \ub \beta) \gamma  + \\
               & (1 + \ub \beta)(1 - \beta - \gamma).\\
    \end{align*}
    The algorithm is $\rho$-competitive in this case if $g \geq 0$ for
    \[
        g := 2 (\rho \OPT - \ALG) =
        -2 - \beta^2 + \beta (2 - 2 \gamma \ub) + \gamma^2 (\ub-1) (\rho-1) + \ub (\rho-1) +
 2 \gamma (\ub + \rho - \ub \rho).
    \]
    The expression $g$ is convex in $\gamma$ as the second derivate is $4(\rho-1)(\ub -1)>0$, hence the adversary chooses the extreme point
    \[
        \gamma = \frac{-\ub + \beta \ub - \rho + \ub \rho}
                        {(\rho-1)(\ub -1)}.
    \]
    The resulting $g$ is concave in $\beta$ as the second derivative is
    \[
        -4 - \frac{4\ub^2}{(\rho-1)(\ub -1)} < 0.
    \]
    Hence the algorithm would like to choose the extreme point
    \[
        \beta = \frac{1 - \ub + \ub^2 - \rho + 2 \ub \rho - \ub^2 \rho}
                {1 - \ub + \ub^2 - \rho + \ub \rho},
    \]
    which is the claimed expression~\eqref{eq:ute:b2}.
    Now $g$ depends solely on $\rho$ and $\ub$ and is increasing in both variables.
    Hence the smallest $\rho$ such that $g \geq 0$ is the root of $g$ in $\rho$ namely
    \begin{equation}  \label{eq:ute:ratio2}
        \rho =
        \frac{-1 - \ub +  2 \ub^2 - \ub^3 + \sqrt{-3 + 6 \ub - 3 \ub^2 - 6 \ub^3 + 10 \ub^4 - 4 \ub^5 + \ub^6}}
            {2 (\ub -1)},
    \end{equation}
    which we would clearly like to simplify.
    Considering the worst upper limit, namely $\ub=\rho$ the ratio simplifies to
    \[
        \rho = \frac{1 + \sqrt{3 + 2 \sqrt{5}}}{2} \approx 1.8668.
    \]

% ----------------------------------------------------------------------------------------------------------
\paragraph{Case $\ub \geq p^*$.}
In this case $\beta\leq 0$ and \ute first tests and postpones the first $1 - \gamma$ fraction of jobs (all of length $\ub$) and then tests and executes the remaining $\gamma$ fraction (all of length 0). Thus the $n^2$ dependent cost for the algorithm is
\[
\ALG= \gamma^2/2 + \ub (1-\gamma)^2 /2 + (1-\gamma)\gamma + (1- \gamma),
\]
while the optimal cost is as in the previous case.

The ratio is at most $\rho$ if $g \geq 0$ for
\[
    g := 2(\rho \OPT - \ALG) =
        \gamma^2 (\ub - 1) (\rho - 1) + \ub (\rho - 1) + 2 \gamma (\ub + \rho - \ub \rho) - 1,
\]
where we used the factor $2$ to obtain a simpler expression.
The expression $g$ is increasing in $\ub$ as its derivative is $(1-\gamma)^2(\rho - 1)>0$.  Therefore we can assume for the worst case $\ub=p^*$.  Now we observe that $g$ is convex in $\gamma$ as the second derivative is $1+\sqrt{4 \rho -3}>0$. Hence the adversary chooses the extreme point for $g$ in $\gamma$, namely
\[
    \gamma = \frac{-1+\sqrt{4 \rho -3}}
                    {1 + \sqrt{4 \rho -3}}.
\]
With these choices of $\ub$ and $\gamma$ the expression $g$ has the form
\[
        g = \frac{3 - 2 (2- \rho) \rho - \sqrt{4 \rho -3 }}
            {2(\rho -1)}.
\]
Evaluated at $\frac{1 + \sqrt{3 + 2 \sqrt{5}}}{2}$ the goal is positive, proving the ratio in this case.

% ----------------------------------------------------------------------------------------------------------
\paragraph{Case $\rho \leq \ub \leq p^*$ and $\gamma \geq 1- \beta$.}
    In this case the algorithm does not postpone the execution of jobs.  The jobs in the first $1- \gamma$ fraction have processing time $\ub$ and the last $\gamma$ fraction jobs have processing time $0$. Therefore the cost of \ute is
    \[
        \ALG= (\ub + 1)(1 - \gamma)^2 /2 +
                \gamma^2/2 +
                (\ub + 1) \gamma (1 - \gamma).
    \]
    In this case $g$ is
    \[
        g := 2 (\rho \OPT - \ALG) = -1 - (1-\gamma^2) \ub + (\ub - (2-\gamma) \gamma (\ub-1)) \rho.
    \]
    The value of $\beta$ is maximized at $\ub=\rho$, which is approximately $\beta^* := 0.2869$.
    We observe that the derivative of $g$ in $\ub$ is negative in the range $\gamma\in[1-\beta^*,1]$, hence $g$ is minimized at $\ub=p^*$.  For this choice $g$ has the approximate form
    \[
        1.4235 + \gamma (-6.7057 + 6.1489 \gamma)
    \]
    which can never become negative, even in the range $\gamma\in[0,1]$.  Therefore we have shown that the ratio is at most $\rho$ also in this last case.

% ----------------------------------------------------------------------------------------------------------
\paragraph{Analysis of the $n$ dependent parts of the costs.}
    Again we consider the same 3 cases as before.
    \begin{description}
        \item[Case $\rho \leq \ub \leq p^*$ and $\gamma \leq 1-\beta$:]
        Here the $n$ dependent costs of \ute and \OPT are
        \begin{align*}
            \ALG&= (\ub + 1)\beta/2 + \gamma/2 + \ub(1- \beta - \gamma)/2
            \\
            \OPT  &= \gamma/2 + \ub(1- \gamma)/2.
        \end{align*}
        The ratio is at most $\rho$ if $g \geq 0$ for
        \[
            g := \rho \OPT - \ALG\approx 0.7309 - 0.4232 \gamma
        \]
        which is positive for all $\gamma\in[0,1]$.
        \item[Case $\ub \geq p^*$:]
        The $n$ dependent cost of \ute is
        \[
            \ALG= \gamma/2 + (\ub+1)(1 - \gamma)/2
        \]
        leading to
        \[
            \rho \OPT - \ALG\approx 0.7118 - 0.2784 \gamma
        \]
        which again is always positive.
        \item[Case $\rho \leq \ub \leq p^*$ and $\gamma \geq 1- \beta$:]
        This time we have
        \[
            \ALG= (\ub+1)(1 - \gamma)/2 + \gamma/2
        \]
        and
        \[
            \rho \OPT - \ALG\approx 0.3508 + 0.0768 \gamma
        \]
        which completes the proof.
    \end{description}
    \qed
\end{proof}

The deterministic lower bound $1.8546$ in Theorem~\ref{thm:detLB} uses the upper limit $\ub\approx 1.9896$. Plugging this choice of $\ub$ into the expression~\eqref{eq:ute:ratio2} shows that \ute has a near-optimal competitive ratio.
\begin{corollary}\label{cor:ute}
	\ute has competitive ratio $\mratio \approx 1.8552$ for scheduling with testing to minimize the sum of completion times for instances with upper limits $\ub\approx 1.9896$. 
\end{corollary}

%!TEX root = main.tex

\section{Optimal Testing for Minimizing the Makespan}\label{sec:makespan}
We consider scheduling with testing with the objective of minimizing the makespan, i.e., the completion time of the last job that is completed. This objective function is special, as the time each job spends on the machine has a linear contribution to the objective function value. This yields that for any algorithm that treats each job independent of the position where it occurs in the schedule, there is a worst-case instance containing only a single job.
\begin{lemma}\label{lem:makespan_add}
	If an algorithm that treats each job independent of the position where it occurs in the schedule is $\mratio$-competitive for one-job instances, it is $\mratio$-competitive also for general instances.
\end{lemma}
\begin{proof}
	Let an instance $I$ with $n$ jobs $j_1, \dots, j_n$ and an arbitrary algorithm as in the statement of the lemma be given. Then the makespan $ALG(I)$ equals the sum of the makespans, if we split the instance into one-job instances. By assumption, the algorithm is $\mratio$-competitive for each one-job instance. Thus, we have
	\begin{align*}
		\ALG(I) & = \sum_{i = 1}^n \ALG(\{j_i\}) \leq \sum_{ i = 1}^n \mratio \cdot \OPT(\{j_i\}) = \mratio \cdot \OPT(I).
	\end{align*}
	\qed
\end{proof}

\paragraph{Deterministic algorithms.}
We apply Lemma~\ref{lem:makespan_add} to give a deterministic algorithm with competitive ratio~$\mratio = \varphi$, the golden ratio, and show this is best-possible.
\begin{theorem}
	Let $\varphi \approx 1.618$ be the golden ratio. Testing each job $j$ if and only if $\ub_j > \varphi$ is an algorithm with competitive ratio $\varphi$ for scheduling with testing to minimize the makespan. This is best possible for deterministic algorithms.
\end{theorem}
\begin{proof}
	By Lemma~\ref{lem:makespan_add} we just need to consider an instance consisting of a single job. Let that job have upper limit~$\ub$ and processing time~$p$. If the algorithm does not test the job, then $\ub\le\varphi$. If $\ub\le 1$, the optimal schedule also executes the job untested, and the competitive
ratio is~$1$. If $\ub>1$, the makespan of the algorithm is $\ub\le \varphi$ and the optimal makespan is at least~$1$, because the optimal makespan is minimized if the job is tested in the optimal schedule and reveals $p = 0$. Thus, the ratio is at most~$\varphi$.

If the algorithm tests the job, then its makespan is $1 + p$, while the optimal makespan is $\min \{\ub, 1 + p\}$. In the worst case, the job has processing time $p = \ub$. Then the ratio is $(1 + \ub)/\ub$, which decreases when the upper limit $\ub$ increases. Thus, it is at most $(1 + \varphi)/\varphi=\varphi$.

To show this is best-possible, consider an instance with a single job with upper limit $\varphi$. Any algorithm that does not test this job has competitive ratio at least $\varphi$, as the optimal makespan is~$1$ if the job has processing time~$0$. Any other algorithm tests the job. If the job has processing time~$\varphi$, the competitive ratio is $(1 + \varphi)/\varphi = \varphi$. %\nic{Don't we use this lower bound somewhere earlier for sum $C_j$?}
\qed
\end{proof}
This shows that there is an algorithm that approaches the optimal execution time up to a factor
$\varphi$. However, this does not lead to a $\varphi$-approximation for the problem of minimizing the sum of completion times where the additional difficulty lies in determining the job order. 

\paragraph{Randomized algorithms.}
For randomized algorithms, we first show that no randomized (or deterministic) algorithm can have competitive ratio $\mratio < 4/3$.
\begin{theorem}
	No algorithm has competitive ratio $\mratio < 4/3$ for minimizing the makespan (resp.\ the sum of completion times) for scheduling with testing.
\end{theorem}
\begin{proof}
	We want to apply Yao's principle~\cite{yao1977probabilistic} and give a randomized instance for which no deterministic algorithm is better than $4/3$-competitive. Consider a one-job instance with $\ub = 2$. Let the job have $p = 0$ and $p = 2$ each with probability $0.5$. The deterministic algorithm that does not test the job has expected makespan $2$ and the deterministic algorithm testing the job also has expected makespan $2$. The expected optimal makespan is $3/2$. Thus, the instance yields the desired bound.
	\qed
\end{proof}
For minimizing the makespan the order in which jobs are treated is irrelevant by Lemma~\ref{lem:makespan_add}. Thus, the only decision an algorithm has to take is whether to test a job. Consider a job with upper limit~$\ub$. We show that the algorithm that executes the job untested if $\ub\le 1$ and otherwise tests it with probability $1 - 1/(\ub^2 - \ub +1)$ is best-possible.
\begin{theorem}
	Our randomized algorithm testing each job with $\ub> 1$ with probability $1 - 1/(\ub^2 - \ub +1)$ has competitive ratio $4/3$ for scheduling with testing to minimize the makespan. This is best-possible.
\end{theorem}
\begin{proof}
	By Lemma~\ref{lem:makespan_add} we just need to consider an instance consisting of a single job.
	If its upper limit~$\ub$ satisfies $\ub\le 1$, the algorithm executes
	the job untested, which is optimal. Therefore, assume for the
	rest of the proof that $\ub> 1$.

	Note that Proposition~\ref{prop:change-indep-deltap}, which
	was stated in the context of minimizing the sum of completion times,
	holds also for single-job instances where the objective is the makespan,
	because for one job the two objectives are the same.
	If $0<p<\ub-1$, we observe that the optimal makespan
	and the expected makespan of the algorithm depend linearly
	on~$p$, so by Proposition~\ref{prop:change-indep-deltap}
	we can set $p$ to $0$ or $\ub-1$ without decreasing the competitive
	ratio. Now, if $\ub-1\le p<\ub$, observe that increasing $p$ to
	$\ub$ increases the expected makespan of the algorithm but
	does not affect the optimum. Therefore, we can assume that
	$p\in\{0,\ub\}$ in a worst-case instance.
	%Furthermore, we can apply Proposition~\ref{prop:change-indep-deltap} \tom{The change of the optimum is not linear as the function changes at $p=\ub-1$, so Proposition~\ref{prop:change-indep-deltap} does not apply directly.}from the section on the sum of completion times objective, as for one job the two objective functions are the same. Thus, the job has processing time $\ub$ or $0$ in a worst-case instance.

	Let us first consider the case $p = \ub$. Then the optimal solution schedules this job without test. Thus, the ratio of algorithm length over optimal length is
	\[ \mratio = \frac{E[\ALG]}{\OPT} = \left(1 - \frac{1}{\bar p^2 - \bar p +1}\right)\frac{\bar p + 1}{\bar p} + \frac{1}{\bar p^2 - \bar p +1}		%=	 \frac{(\bar p^2 - \bar p)(\bar p +1)+ \bar p}{\bar p(\bar p^2 - \bar p +1)}		=	\frac{\bar p^3 - \bar p + \bar p}{\bar p(\bar p^2 - \bar p +1)}
	= \frac{\bar p^2}{\bar p^2 - \bar p +1}.\]
	Otherwise, we have $p=0$. Then  we have
		\[ \mratio = \frac{E[\ALG]}{\OPT} = \left(1 -\frac{1}{\bar p^2 - \bar p +1} \right) + \frac{1}{\bar p^2 - \bar p +1}\bar p		=	 \frac{\bar p^2}{\bar p^2 - \bar p +1}.\]
	This function is maximized at $\bar p = 2$, which yields the competitive ratio $4/3$.
	\qed
\end{proof}

%!TEX root = notes.tex

\section{Conclusion}
In this paper we have introduced an adversarial model of scheduling with
testing where a test can shorten a job but the time for the test also
prolongs the schedule, thus making it difficult for an algorithm to find
the right balance between tests and executions. We have presented upper and
lower bounds on the competitive ratio of deterministic and randomized algorithms
for a single-machine scheduling problem with the objective of minimizing the
sum of completion times or the makespan. An immediate open question is whether
it is possible to achieve competitive ratio below~$2$ for minimizing the
sum of completion times with a deterministic algorithm for arbitrary instances.
Further interesting directions for future work include
the consideration of job-dependent test times or other scheduling problems
such as parallel machine scheduling or flow shop problems. More generally, the study
of problems with explorable uncertainty in settings where the costs for
querying uncertain data directly contribute to the objective value is
a promising direction for future work.

\section*{Acknowledgements}
% \begin{acknowledgements}
We would like to thank Markus Jablonka and Bruno Gaujal for helpful discussions about
the algorithm \TA, as well as an anonymous referee for pointing us to related work on exploration versus exploitation in the multi-armed bandit framework.
% \end{acknowledgements}

\bibliographystyle{abbrv}
\bibliography{uncertainty}

\end{document}